\documentclass[11pt]{amsart}
\usepackage[foot]{amsaddr}
\usepackage{ifxetex}
\ifxetex
  \usepackage[no-math]{fontspec}
\else
\fi
\usepackage{amsmath}
\usepackage{amsfonts}
\usepackage{amssymb}
\usepackage{amsthm}
\usepackage{fullpage}
\usepackage[lining,semibold,type1]{libertine} % a bit lighter than Times--no osf in math
\usepackage[T1]{fontenc} % best for Western European languages
\usepackage{textcomp} % required to get special symbols
\usepackage[varqu,varl]{inconsolata}% a typewriter font must be defined
\usepackage[libertine,vvarbb]{newtxmath}
\usepackage[scr=rsfso]{mathalfa}
\usepackage{bm}

\usepackage{listings}
\usepackage{hyperref, color}
\hypersetup{colorlinks=true,citecolor=blue, linkcolor=blue, urlcolor=blue}
\usepackage[linesnumbered,boxed,ruled,vlined]{algorithm2e}
\usepackage{bm}
\usepackage{bbm}
\usepackage[numbers]{natbib}
\usepackage{xcolor}
\usepackage{enumerate} 
\usepackage{enumitem}
\usepackage{tabularx}
\usepackage{array}
\usepackage{cleveref}
\usepackage{mathrsfs}

\newcolumntype{L}[1]{>{\raggedright\arraybackslash}p{#1}}
\newcolumntype{C}[1]{>{\centering\arraybackslash}m{#1}}
\newcolumntype{R}[1]{>{\raggedleft\arraybackslash}p{#1}}
  
\usepackage{makecell}

\usepackage{footnote}
\makesavenoteenv{tabular}

\newcommand{\DTV}[2]{d_{\mathrm{TV}}\left({#1},{#2}\right)}

\newcommand{\e}{\mathrm{e}}

\renewcommand{\epsilon}{\varepsilon}

\newcommand{\tmix}{T_{\textsf{mix}}}

\newcommand{\Xvtx}{X^{\textsf{vtx}}}
\newcommand{\Xcol}{X^{\textsf{val}}}
\newcommand{\Yvtx}{Y^{\textsf{vtx}}}
\newcommand{\Ycol}{Y^{\textsf{val}}}

\newtheorem{theorem}{Theorem}[section]
\newtheorem{observation}[theorem]{Observation}

\newtheorem*{claim*}{Claim}
\newtheorem{condition}[theorem]{Condition}

\newtheorem{lemma}[theorem]{Lemma}
\newtheorem{proposition}[theorem]{Proposition}
\newtheorem{corollary}[theorem]{Corollary}
\theoremstyle{definition}

\newtheorem{definition}[theorem]{Definition}

\newtheorem*{remark*}{Remark}

\def\Pr{\mathop{\mathbf{Pr}}\nolimits}

\newcommand{\norm}[1]{\left\Vert#1\right\Vert}
\newcommand{\set}[1]{\left\{#1\right\}}
 \newcommand{\tuple}[1]{\left(#1\right)} \newcommand{\eps}{\varepsilon}
 
 \newcommand{\tp}{\tuple}

 \newcommand{\pin}{\mathsf{Pin}}
 
\newcommand{\defeq}{\triangleq}

\newcommand{\abs}[1]{\left\vert#1\right\vert}

\def\*#1{\boldsymbol{#1}} % Use \*A for \mathbf{A}
\def\+#1{\mathcal{#1}} % Use \+A for \mathcal{A}
\def\-#1{\mathrm{#1}} % Use \-A for \mathrm{A}
\def\^#1{\mathscr{#1}} % Use \^A for \mathscr{A}

\usepackage{todonotes}

\usepackage{xifthen}

\renewcommand{\Pr}[2][]{ \ifthenelse{\isempty{#1}}
  {\mathbf{Pr}\left[#2\right]} {\mathbf{Pr}_{#1}\left[#2\right]} } % Use \Pr[a]{b} for \mathbf{Pr}_a[b], \Pr{b} for  \mathbf{Pr}[b]
\newcommand{\E}[2][]{ \ifthenelse{\isempty{#1}}
  {\mathbf{\mathbf{E}}\left[#2\right]}
  {\mathbf{\mathbf{E}}_{#1}\left[#2\right]} }
  \newcommand{\Var}[2][]{ \ifthenelse{\isempty{#1}}
  {\mathbf{\mathbf{Var}}\left[#2\right]}
  {\mathbf{\mathbf{Var}}_{#1}\left[#2\right]} }

\newcommand{\one}[1]{\textbf{1}\left[#1\right]}

\title{Rapid Mixing from Spectral Independence beyond the Boolean Domain}
\date{}
\author{Weiming Feng}
\author{Heng Guo}
\author{Yitong Yin}
\author{Chihao Zhang}

\address[Weiming Feng, Yitong Yin]{State Key Laboratory for Novel Software Technology, Nanjing University, 163 Xianlin Avenue, Nanjing, Jiangsu Province, China. \textnormal{E-mail: \url{fengwm@smail.nju.edu.cn, yinyt@nju.edu.cn}}}
\address[Heng Guo]{School of Informatics, University of Edinburgh, Informatics Forum, Edinburgh, EH8 9AB, United Kingdom. \textnormal{E-mail: \url{hguo@inf.ed.ac.uk}}}
\address[Chihao Zhang]{John Hopcroft Center for Computer Science, Shanghai Jiao Tong University, 800
  Dongchuan Road, Minhang District, Shanghai, China. \textnormal{E-mail: \url{chihao@sjtu.edu.cn}}}
\begin{document}
\thanks{This research was supported by the National Key R\&D Program of China 2018YFB1003202 and the NSFC Nos. 61722207, 61672275 and 61902241.}

\begin{abstract}
  We extend the notion of spectral independence (introduced by Anari, Liu, and Oveis Gharan \cite{anari2020spectral}) from the Boolean domain to general discrete domains.
  This property characterises distributions with limited correlations, and implies that the corresponding Glauber dynamics is rapidly mixing.

%Spectral independence was introduced in~\cite{anari2020spectral} to characterise % the property that the correlation matrices of a Boolean domain distribution and its conditional distributions are of bounded maximum eigenvalue.
%distributions with limited correlation over the Boolean domain.
%
%We extend this notion to general discrete domains and show its implications to the mixing time of Glauber dynamics. 
%
%We then show that the recursive coupling technique of~\cite{GKM15} can be used to establish spectral independence for multi-spin systems and therefore this implies efficient sampling algorithms for their Gibbs distributions. 
%
As a concrete application, we show that Glauber dynamics for sampling proper $q$-colourings mixes in polynomial-time for the family of triangle-free graphs with maximum degree $\Delta$ provided $q\ge (\alpha^*+\delta)\Delta$ where $\alpha^*\approx 1.763$ is the unique solution to $\alpha^*=\exp\tp{1/\alpha^*}$ and $\delta>0$ is any constant. 
This is the first efficient algorithm for sampling proper $q$-colourings in this regime with possibly unbounded $\Delta$.
Our main tool of establishing spectral independence is the recursive coupling by Goldberg, Martin, and Paterson~\cite{GMP05}.
\end{abstract}

\maketitle

\section{Introduction}

Let $V$ be a set of variables, each of which takes values from a discrete domain of size $q\ge 2$.
Sampling from a complicated joint distribution $\mu$ over the state space $[q]^V=\set{0,1,\dots,q-1}^V$ is an important yet intricate computational task. 
The \emph{Markov chain Monte Carlo (MCMC)} method is the most powerful and flexible technique to design efficient samplers. 
We will focus on \emph{Glauber dynamics} in this paper, which is one of the simplest and most widely used Markov chains.
In each step, it does the following:
\begin{enumerate}
  \item choose a variable uniformly at random;
  \item resample the value of the variable according to its marginal distribution conditioned on the values of all other variables.
\end{enumerate}
Denote by $\mu_t$ the distribution of the state after $t$ steps. 
It is usually straightforward to show that $\mu_t$ converges to the desired distribution $\mu$ as $t$ tends to $\infty$. 
However, the more challenging task is to understand \emph{how fast} the distance between $\mu_t$ and $\mu$ converges to $0$.
This rate of convergence is known as the \emph{mixing time}. 
Many tools have been invented towards proving fast convergence, 
or the so-called rapid mixing property of Glauber dynamics. 
We refer the reader to~\cite{levin2017markov} for a recent monograph on this topic.

%The difficulty to analyse the convergence rate of Glauber dynamics mainly comes from the rich and complicated landscape of the distributions.
%
%In a recent breakthrough~\cite{LGV19}, who proved that the Glauber dynamics for sampling bases of matroids are efficient, brought new insights to tackle the difficulties. 
%
%In their proof, the Glauber dynamics has been interpreted as a random walk on a high-dimensional cube, and its convergence has been reduced to the convergence of random walks on spaces with low-dimension, which is much easier to analyse. 
%
%The idea has been further developed in the work of~\cite{anari2020spectral,chen2020rapid}, in which the authors showed that the fast convergence of ``low-dimensional walks'' have natural connections with the notion of ``spatial mixing'', a key ingredient in many previous efficient samplers, for a family of special joint distributions called the Gibbs distributions of ``two-spin systems''.
%

Distributions of interest often have rich and complicated landscapes,
which makes analysing the convergence rate of Glauber dynamics a long-standing challenge in theoretical computer science.
To tackle this challenge, various techniques were introduced, such as canonical paths \cite{jerrum1989approximating} and path coupling \cite{bubley1997path}.
In a more recent line of work \cite{DK17,Opp18,KO20,alev2020improved}, 
an interesting new method of analysing the mixing time emerged via the so-called ``local-to-global'' argument for high-dimensional expanders.
This technique has played a central role in a few recent breakthrough results,
such as uniform sampling of matroid bases \cite{ALOV19,CGM19,ALOV20},\footnote{The bases exchange walk for matroids can be viewed as Glauber dynamics as follows. Consider $r$ variables, where $r$ is the rank of the matroid. Each variable can take values from the ground set subject to the matroid constraint. Each bases exchange move is exactly resampling the value of a randomly chosen variable conditioned on the assignment of all other variables.}
and a tight analysis for the hardcore model \cite{anari2020spectral} and more generally for anti-ferromagnetic 2-spin systems \cite{chen2020rapid}.

%\begin{algorithm}[ht]
%\SetKwInOut{Input}{Input}
%\SetKwInOut{Output}{Output}
%\Input{a distribution $\mu$ with support $\Omega \subseteq [q]^V$, and a threshold $T$;}
%\Output{a random configuration $X \in \Omega$;}
%let $X \in \Omega$ be an arbitrary feasible configuration\;
%\For{each $t$ from $1$ to $T$}{
%pick a variable $v \in V$ uniformly at random\;
%resample $X_v$ from the marginal distribution $\mu_v^{X_{V \setminus \{v\}}}$\label{line-resmaple}\;
% }
%\Return{$X$}
%\caption{Glauber dynamics}\label{alg-gd}
%\end{algorithm}

Of particular interest to us is the work of Anari, Liu, and Oveis Gharan \cite{anari2020spectral}.
In order to apply the result of Alev and Lau \cite{alev2020improved}, they introduced \emph{spectral independence},
which is the property that the correlation matrix of $\mu$ and all of its conditional distributions have bounded maximum eigenvalues.
They focused on the $q=2$ case. 
% 
%For example, for any undirected graph $G=(V,E)$, we can let $q=2$ and use $\mu$ to denote the uniform distribution over all independent sets on $G$. It is known that an efficient sampling algorithm, or formally a \emph{fully polynomial-time approximate uniform sampler}, does not exists for these $\mu$ unless $\mathbf{NP}=\mathbf{RP}$~\cite{}. 
Formally, for each feasible\footnote{A configuration $\sigma \in \{0,1\}^V$ is \emph{feasible} if $\mu(\sigma) > 0$. A partial configuration $\sigma_\Lambda \subseteq \{0,1\}^\Lambda$ for $\Lambda \subseteq V$ is \emph{feasible} if it can be extended to a feasible configuration.} 
$\sigma_\Lambda \in \{0,1\}^\Lambda$ where $\Lambda \subseteq V$, Anari, Liu and Oveis Gharan defined a \emph{signed pairwise influence matrix} $I_\mu^{\sigma_{\Lambda}}$ by $I_\mu^{\sigma_\Lambda}(u,v) \triangleq (\mu_v^{\sigma_\Lambda, u \gets 1}(1) - \mu_v^{\sigma_\Lambda, u \gets 0}(1)) \cdot \one{u \neq v}$ for all $u, v \in V \setminus \Lambda$, where $\mu_v^{\sigma_\Lambda, u \gets i} (i = 0, 1)$ is the marginal distribution on $v$ induced from $\mu$ conditional on the configuration on $\Lambda$ fixed as $\sigma_\Lambda$ and that $u$ is fixed to $i$.
In~\cite{anari2020spectral}, a distribution $\mu$ over $\{0,1\}^V$ is said to be spectrally independent if 
for any $\Lambda \subseteq V$, any feasible $\sigma_\Lambda \in \{0,1\}^\Lambda$,
the maximum eigenvalue $\lambda_{\max}(I^{\sigma_\Lambda}_\mu)$ can be upper bounded appropriately. 
They proved that the spectral independence property implies rapid mixing of Glauber dynamics. 
Using this tool, they confirmed a long-standing conjecture: Glauber dynamics for the Gibbs distribution of the hardcore model is rapidly mixing up to the uniqueness threshold. 
Later on, Chen, Liu and Vigoda~\cite{chen2020rapid} further extended the mixing results to general antiferromagnetic 2-spin systems. 

%So far, the spectral independence is only established for distributions over $\{0,1\}^V$.
%
Despite the success in the Boolean domain, the machinery developed by Anari, Liu and Oveis Gharan does not handle many important distributions, 
such as the Gibbs distribution of Potts models where $q>2$ can be any positive integer. 
Therefore a natural question is whether the approach developed in~\cite{anari2020spectral,chen2020rapid}, 
or more specifically the notion of spectral independence, can be generalised beyond the Boolean domain. 
We note two interconnected difficulties for this task: 
(1) when $q>2$, there are many non-equivalent choices for the definition of influence between two variables $u,v\in V$; 
and (2) it is not clear whether the elegant connection~\cite[Theorem 3.1]{anari2020spectral} between the ``local'' random walks of \cite{alev2020improved} and the spectrum of the influence matrix still holds beyond the Boolean domain.

Our first contribution is to introduce the following generalised influence matrix. 
This definition allows us to recover the part relevant to rapid mixing in the aforementioned result~\cite[Theorem 3.1]{anari2020spectral} for the more general setting.

\begin{definition}[Influence Matrix]\label{def-inf-matrix}
%\ytodo{ The definition requires the notion of feasibility and a notation of marginal distribution.}
Let $\mu$ be a distribution over $[q]^V$. 
Fix any $\Lambda \subseteq V$ and any feasible $\sigma_{\Lambda} \in [q]^{{\Lambda}}$.
For any distinct $u,v\in V\setminus \Lambda$, we define the \emph{(pairwise) influence} of $u$ on $v$ by 
%$\Psi^{\sigma_\Lambda}_{\+I}$ under the pinning $\sigma$. For all $u \in \Lambda$, $\Psi_{\sigma}(u,u) \triangleq 0$; and for all $u, v \in \Lambda$ with $u \neq v$,
%\begin{align}
%\label{eq-def-psi}
%\Psi^{\sigma_\Lambda}_{\mu}(u,v) \triangleq 
%\begin{cases}
%\max_{i,j\in[q]:\,\, \mu_u^{\sigma_\Lambda}(i),\mu_u^{\sigma_\Lambda}(j)>0}\DTV{\mu_{v}^{\sigma,u\gets i} }{ \mu_{v}^{\sigma,u\gets j}} & \text{if }u\neq v,\\
%0 & \text{if }u=v.
%\end{cases}
%\end{align}
\begin{align}
\label{eq-def-psi}
\Psi^{\sigma_\Lambda}_{\mu}(u,v) \triangleq 
%\max_{i,j\in[q]\in[q]:\,\, \mu_u^{\sigma_\Lambda}(i),\mu_u^{\sigma_\Lambda}(j)>0}\DTV{\mu_{v}^{\sigma,u\gets i} }{ \mu_{v}^{\sigma,u\gets j}}, 
\max_{i,j\in\Omega_u^{\sigma_{\Lambda}}}\DTV{\mu_{v}^{\sigma_{\Lambda},u\gets i} }{ \mu_{v}^{\sigma_{\Lambda},u\gets j}}, 
\end{align}
where $\Omega_u^{\sigma_{\Lambda}}\triangleq\left\{i\in[q]\mid \mu_u^{\sigma_\Lambda}(i)>0\right\}$ denotes the set of possible values of $u$ given condition $\sigma_\Lambda$,
$\DTV{\cdot}{\cdot}$ denotes the total variation distance between two distributions,
and for $c=i$ or $j$, $\mu_v^{\sigma_\Lambda, u \gets c}$ is the marginal distribution on $v$ induced from $\mu$ conditional on the configuration on $\Lambda$ fixed as $\sigma_\Lambda$ and that $u$ is fixed to $c$. 

Furthermore, let $\Psi^{\sigma_\Lambda}_{\mu}(u,v)\triangleq 0$ for $u=v$ and write $\Psi^{\sigma_\Lambda}_{\mu}$ for the \emph{(pairwise) influence matrix} whose entries are given by $\Psi^{\sigma_\Lambda}_{\mu}(u,v)$.
\end{definition}

In our definition, $\Psi^{\sigma_\Lambda}_{\mu}(u,v)$ is the maximum influence on $v$ caused by a single disagreement on $u$ conditional on $\sigma_\Lambda$. 
The entries of $\Psi^{\sigma_\Lambda}_{\mu}$ are total variation distances and are therefore non-negative. 
We remark that our definition is not identical to the original influence matrix $I^{\sigma_\Lambda}_\mu$ in~\cite{anari2020spectral} even in the Boolean domain since the latter is signed. 
Nevertheless, if $q=2$, it holds that $\Psi^{\sigma_\Lambda}_{\mu}(u,v) = \abs{I^{\sigma_\Lambda}_\mu(u,v)}$.

With the definition of the influence matrix, we define spectral independence for general $q\ge 1$ as follows.

\begin{definition}[Spectral Independence]
\label{def:spectral-independence}
%We say a distribution $\mu$ over $[q]^V$, where $n=|V|$, 
%is \emph{$(\eta_0,\ldots,\eta_{n-2})$-spectrally independent}, 
%if for every $0\le k\le n-2$, $\Lambda\subseteq V$ of size $k$ and any feasible $\sigma_{\Lambda}\in[q]^\Lambda$, one has the spectral radius $\rho(\Psi^{\sigma_{\Lambda}}_{\mu})\le\eta_k$.
We say a distribution $\mu$ over $[q]^V$, where $n=|V|$, is \emph{$(C,\eta)$-spectrally independent}, 
if %for any $\Lambda\subseteq V$ of size at most $n-2$ and
every $0\le k\le n-2$, $\Lambda\subseteq V$ of size $k$ and 
any feasible $\sigma_{\Lambda}\in[q]^\Lambda$,
the spectral radius $\rho\left(\Psi^{\sigma_{\Lambda}}_{\mu}\right)$ of the influence matrix $\Psi^{\sigma_{\Lambda}}_{\mu}$ satisfies
\begin{align*}
\rho\left(\Psi^{\sigma_{\Lambda}}_{\mu}\right) \leq  C \quad\text{and}\quad  
%\frac{\rho\left(\Psi^{\sigma_{\Lambda}}_{\mu}\right)}{|V|-|\Lambda|-1} \leq \eta.
\frac{\rho\left(\Psi^{\sigma_{\Lambda}}_{\mu}\right)}{n-k-1} \leq \eta.
\end{align*}
\end{definition}

%
%In order to state your main theorem, we now formally define the Glauber dynamics for general joint distributions in \Cref{alg-gd}.
%
Consider the Glauber dynamics for a general distribution $\mu$ and let $P_{\mathsf{Glauber}} \in \mathbb{R}_{\geq 0}^{\Omega\times\Omega}$  be its transition matrix.
It is well-known that the Glauber dynamics converges to stationary distribution $\mu$ when $P_{\mathsf{Glauber}}$ is irreducible, see e.g.~\cite{levin2017markov}.
%Let $\mu$ be a distribution over $[q]^V$ and $\Omega \triangleq \{\sigma \in [q]^V \mid \mu(\sigma) > 0\}$ the support of $\mu$.
%Each $X \in \Omega$ is called a feasible \emph{configuration}.
%The \emph{Glauber dynamics}~(\Cref{alg-gd}) is one of the most fundamental \emph{Markov chain Monte Carlo (MCMC)} approach to sample from complex distributions. 
%to sample random configurations from 

%The notation $\mu_v^{X_{V \setminus \{v\}}}$ in \Cref{line-resmaple} denotes the marginal distribution on $v$ conditional on the configuration on $V \setminus \{v\}$ is fixed as $X_{V \setminus \{v\}}$. 
%The distribution $\mu_v^{X_{V \setminus \{v\}}}$ is well-defined because $X$ is always feasible. 

%The Glauber dynamics is a random walk over the configuration space $\Omega$.
%
%We use $P_{\mathsf{Glauber}} \in \mathbb{R}_{\geq 0}^{\Omega\times\Omega}$ to denote the transition matrix of Glauber dynamics.
%
The rate of convergence of Glauber dynamics is captured by the \emph{mixing time}, defined as:
%To set a proper time limit $T$,
%it is crucial to understand how fast Glauber dynamics converges to the stationary distribution $\mu$, which is captured by the notion of \emph{mixing time}
\begin{align*}
\forall\,  0< \epsilon < 1, \quad \tmix(\epsilon) = \max_{\*x_0 \in \Omega} \min\left\{t \mid \DTV{P_{\mathsf{Glauber}}^t(\*x_0,\cdot)}{\mu} \leq \epsilon \right\}.
\end{align*}

Our main theorem states that the Glauber dynamics for $\mu$ is rapidly mixing if $\mu$ is spectrally independent.
\begin{theorem}
\label{theorem-main}
%Let $\+I=(G,[q],\*b,\*A)$ be a \emph{spin system} with Gibbs distribution $\mu$, where $G = (V, E)$ and $n = |V|$.
%If $\mu$ is $(\eta_0,\ldots,\eta_{n-2})$-spectrally independent, and there exist $C>0$ and $0<\eta<1$ such that $\eta_k\le C$ and $\frac{\eta_k}{n-k-1}\le \eta$ for every $0\le k\le n-2$,
Let $\mu$ be a distribution over $[q]^V$. 
If $\mu$ is $(C,\eta)$-spectrally independent for $C\geq0$ and $0\leq\eta<1$,
then the Glauber dynamics for $\mu$ has mixing time
\begin{align*}
\tmix(\epsilon) \leq \frac{n^{1+2C}}{(1-\eta)^{2+2C}}\left(  \log \frac{1}{\epsilon \mu_{\min}} \right),
\end{align*}
where %$\tmix(\epsilon) = \max_{X_0 \in \Omega} \min\left\{t \mid \DTV{P_{\mathsf{Glauber}}^t(X_0,\cdot)}{\mu} \leq \epsilon \right\}$ and 
$n=|V|$ and $\mu_{\min}\triangleq \min\{\mu(\sigma)\mid\sigma\in[q]^V\land \mu(\sigma)>0\}$.
\end{theorem}

This generalises a similar result by Anari, Liu and Oveis Gharan~\cite{anari2020spectral} for $q=2$. 
Their proof is based on a linear algebra argument which completely characterises the spectrum of their influence matrix in terms of the spectrum of the local random walks,
so that the result of Alev and Lau \cite{alev2020improved} applies.
However, it is not clear whether a similar argument exists for general $q$. 
Instead our main contribution is a new coupling based argument to connect spectral independence to rapid mixing of Glauber dynamics, which holds for any $q\in\mathbb{N}$. 
To be more specific, we also utilises the result of Alev and Lau \cite{alev2020improved}.
We show that the second largest eigenvalue of the local random walk can be bounded in terms of the spectral radius of our influence matrix (see \Cref{lemma-decay-imples-lambda2}).
In order to relate these two quantities, we employed a coupling analysis reminiscent of the work of Hayes \cite{hayes2006simple}.
See \Cref{section-overview} for an overview of our proof.

To apply our result, one needs to verify the spectral independence property, which is equivalent to bound the spectral radius of an influence matrix. 
This is not an easy task in general.
A more tractable way is to bound the induced $1$-norm or the induced $\infty$-norm of the influence matrix, which are upper bounds of its spectral radius.
\begin{corollary}
\label{corollary-main}
Let $\mu$ be a distribution over $[q]^V$, where $n = \abs{V}$. 
If there exist two constants $C \geq 0$ and $0 \leq \eta< 1$ such that
for every $0\le k\le n-2$, $\Lambda\subseteq V$ of size $k$ and any feasible $\sigma_{\Lambda}\in[q]^\Lambda$, the influence matrix $\Psi_{\mu}^{\sigma_\Lambda}$ satisfies one of following two conditions:
\begin{itemize}
\item \textbf{bounded all-to-one influence:}	 
$$\norm{\Psi^{\sigma_{\Lambda}}_\mu}_1 \defeq \max_{v \in V \setminus \Lambda}\sum_{u \in V \setminus \Lambda}\Psi_{\mu}^{\sigma_\Lambda}(u,v) \leq \min \left\{C,\eta(n - k -1) \right\}$$
\item \textbf{bounded one-to-all influence:} 
$$\norm{\Psi^{\sigma_{\Lambda}}_\mu}_\infty \defeq \max_{u \in V \setminus \Lambda}\sum_{v \in V \setminus \Lambda}\Psi_{\mu}^{\sigma_\Lambda}(u,v) \leq \min \left\{C,\eta(n - k -1) \right\}$$
\end{itemize}
%\begin{align}
%\label{eq-condition-influ}
%\norm{\Psi^{\sigma_{\Lambda}}_\mu}_1 \leq \min \left\{C,\eta(n - k -1) \right\}  \quad\text{or}\quad  \norm{\Psi^{\sigma_{\Lambda}}_\mu}_\infty \leq \min \left\{C,\eta(n - k -1) \right\},
%\end{align}
then the Glauber dynamics for $\mu$ has  mixing time
\begin{align*}
\tmix(\epsilon) \leq \frac{n^{1+2C}}{(1-\eta)^{2+2C}}\left(  \log \frac{1}{\epsilon \mu_{\min}} \right),
\end{align*}
where $\mu_{\min}\triangleq \min\{\mu(\sigma)\mid\sigma\in[q]^V\land \mu(\sigma)>0\}$.
\end{corollary}

The conditions in \Cref{corollary-main} have been previously established for the hardcore model \cite{anari2020spectral} (all-to-one influence) and more generally for anti-ferromagnetic 2-spin systems \cite{chen2020rapid} (one-to-all influence).\footnote{Although in \cite{anari2020spectral} and \cite{chen2020rapid}, the corresponding conditions were established for the signed influence matrix $I_\mu^{\sigma_{\Lambda}}$, they are still applicable to our \Cref{corollary-main} since $\|\Psi_\mu^{\sigma_{\Lambda}}\|_1=\|I_\mu^{\sigma_{\Lambda}}\|_1$ and $\|\Psi_\mu^{\sigma_{\Lambda}}\|_{\infty}=\|I_\mu^{\sigma_{\Lambda}}\|_{\infty}$ when $q=2$.}
%To apply it, one needs to bound the total influence on one variable (all-to-one influence), or the total influence from one variable (one-to-all influence).
Such conditions are quite natural for \emph{Gibbs distributions} induced by $q$-\emph{spin systems}.
Roughly speaking, a $q$-spin system is defined on a graph $G = (V,E)$, where vertices represent random variables that take values in $[q]$, and edges model pairwise interactions.  
Both ``bounded all-to-one influence'' and ``bounded one-to-all influence'' can be viewed as some forms of the \emph{spatial mixing} or \emph{correlation decay} property of the $q$-spin systems. 
This property roughly says that the influence between two vertices decays rapidly with respect to their distance in the graph $G$ and has been widely exploited to design efficient samplers for the Gibbs distribution. 
For antiferromagnetic 2-spin systems, the rapid mixing regimes obtained by~\cite{anari2020spectral,chen2020rapid} match the best known correlation decay results~\cite{Wei06,LLY13,guo2018uniqueness,SS20}. 
We show that our notion of spectral independence can also be used to obtain efficient sampling algorithms up to known correlation decay regime for multi-spin systems~\cite{GMP05,GKM15}.

\subsection{Application to spin systems}
As a concrete application, we consider an important multi-spin system, i.e.\ proper graph $q$-colourings, or equivalently the anti-ferromagnetic Potts model with the temperature going to negative infinity.
A graph $q$-colouring instance is specified by $(G,[q])$, where $[q]=\{0,1,\ldots,q-1\}$ is a set of colours and $G=(V,E)$ is a simple undirected graph.
A proper colouring $X \in [q]^V$ assigns each vertex $v \in V$ a colour $X_v \in [q]$ such that $X_u \neq X_v$ for all $\{u,v\}\in E$.
Let $\Omega$ denote the set of all proper colourings and $\mu$ denote the uniform distribution over $\Omega$.
In this concrete setting, the Glauber dynamics works as follows. The chain starts from an arbitrary proper colouring $X \in \Omega$, and in each step, it does:
\begin{enumerate}
  \item pick a vertex $v \in V$ uniformly at random;
  \item update $X_v$ by choosing a colour from $[q] \setminus \{X_u \mid \{v,u\} \in E \}$ uniformly at random.	
\end{enumerate}

When $q\ge \Delta+2$, the chain converges to $\mu$ for any initial colouring $X$.
However, it is a notorious open problem that whether the condition $q\ge \Delta+2$ also guarantees rapid mixing. 
We make some progress towards this problem by proving the following result.

Let $\alpha^* \approx 1.763\ldots$ be the positive root of the equation $x^x = \mathrm{e}$.	
Using \Cref{theorem-main}, we obtain the following.

\begin{theorem}
\label{theorem-colouring}
Let $\delta > 0$ be a constant. For any graph colouring instance $(G,[q])$ where $G$ is triangle-free and $q \geq (\alpha^* + \delta)\Delta$, the Glauber dynamics on $(G,[q])$ has mixing time
\begin{align*}
	\tmix(\eps) \leq \tp{9\e^5 n}^{2 + 9/\delta}\log\tp{\frac{q}{\eps}},
\end{align*}
where $n$ is the number of vertices in $G$ and $\Delta \geq 3$ is the maximum degree of $G$.
\end{theorem}
\noindent
%Note that the sampling problem is trivial if $\Delta \leq 2$.
While \Cref{theorem-colouring} is stated for graph $q$-colouring instances, the mixing time upper bound holds for the more general list colouring problem (see \Cref{thm-2D+1.76D}).
In fact, the same rapid mixing bound holds as long as the marginal probabilities are always appropriately upper bounded.
This is formally stated by Condition~\ref{cond:main} and Theorem~\ref{thm-family-main}.
%
%Such marginal upper bounds based criterion for sampling colouring has appeared in~\cite{liu2019deterministic}, while their running $n^{f(\Delta)}$ where $f(\Delta) = \exp(\mathrm{poly}(\Delta))$.

It is instructive to compare \Cref{theorem-colouring} with the vast body of literature on this problem. 
The study was initiated by the pioneering work of Jerrum~\cite{jerrum1995very} and of Salas and Sokal \cite{salas1997absence},
who showed $O(n\log n)$ mixing time if $q \geq (2+\delta)\Delta$. 
So far, in general graphs, the best result is the $O(n^2)$ mixing time when $q \geq (\frac{11}{6}-\eps_0)\Delta$ for some absolute small constant $\eps_0 > 0$~\cite{vigoda2000improved,chen2019improved}.
For restricted families of graphs, there is a long line of work that studied the mixing time of Glauber dynamics under various conditions~\cite{dyer2001fewer,hayes2003randomly,hayes2003non, GMP05, HV06, molloy2004glauber, hayes2013local, dyer2013randomly}. 
A few results most relevant to \Cref{theorem-colouring} are listed in \Cref{table:results}. 
The triangle-free condition, or more generally the requirement on the girth of the graph, has played an important role to improve the dependency of $q$ and $\Delta$.
For a more complete picture, we refer the reader to the survey~\cite{frieze2007survey}.
\begin{table}[h]
\begin{tabular}{|C{1.6cm}|C{2.9cm}|C{2cm}|C{3cm}|C{3.5cm}|}
\hline   & Regime & Girth & Other requirement & Mixing time $\tmix(\frac{1}{4\mathrm{e}})$ \\
%\hline 
%\cite{dyer2001fewer} & $q \geq (\alpha^* + \delta)\Delta$ & $\Omega(\log\log n)$ & $\Delta = \Omega(\log n)$ & $O(\frac{n}{\delta} \log n)$ \\
\hline 
\cite{GMP05}  & $q > \alpha^*\Delta$ & {$\geq 4$} & {\small {$\Delta=O(1)$  and  neighbourhood amenable}} & $O(n^2)$ \\ 
\hline 
\cite{HV06} & $q \geq (\alpha^* + \delta)\Delta$ & {$\geq 4$} & $\Delta = \Omega(\log n)$ & {$O(\frac{n}{\delta} \log n)$} \\
\hline 
\cite{dyer2013randomly} & {$q \geq (\alpha^* + \delta)\Delta$} & {$\geq 5$} & $\Delta \geq \Delta_0(\delta)$ & {\small $O(\frac{n}{\delta} \log n)$} \\
\hline
This work & {$q \geq (\alpha^* + \delta)\Delta$} & {$\geq 4$} & -- & { $O((9\mathrm{e}^5n)^{2+9/\delta}\log q)$}\\
\hline
\end{tabular}
\vspace{3pt}
\caption{Mixing time results for sampling proper graph $q$-colourings.}\label{table:results}
\end{table}
%\todo{$\Omega$ with $\delta$}
%In addition to results in \Cref{table:results},
%one can also obtain better regimes such as $q \geq (\beta^*+\delta)\Delta$ (where $\beta^* \approx 1.489\ldots$ is the root of $(1 - \mathrm{e}^{-1/\beta})^2 + \beta \mathrm{e}^{-1/\beta} = 1$) and  $q \geq (1+\delta)\Delta$ if we further assume sufficiently large girth and maximum degree ~\cite{molloy2004glauber,dyer2013randomly,hayes2003non}.
%For more details, see~\cite{frieze2007survey,chen2019improved}.

In addition to algorithms based on Glauber dynamics mentioned above, 
using the reduction from sampling to counting~\cite{jerrum1986random}, 
one can obtain sampling algorithms from approximate counting algorithms~\cite{GK12,lu2013improved,liu2019deterministic}. The current best FPTAS for counting $q$-colourings is given by Liu, Sinclair and  Srivastava~\cite{liu2019deterministic}. 
The algorithm has running time $n^{f(\Delta)}$ where $f(\Delta) = \exp(\mathrm{poly}(\Delta))$ in (1)~general graphs with $q \geq 2\Delta$; (2) triangle-free graphs with $q \geq (\alpha^* + \delta)\Delta + \beta(\delta)$. 
Therefore, their algorithm does not run in polynomial-time if $\Delta=\omega(1)$.

%Dyer and Frieze~\cite{dyer2001fewer} proved that for graphs with $\Delta = \Omega(\log n)$ and girth $g = \Omega(\log \log n)$, the Glauber dynamics mixes in $O(n \log n)$ steps if $q \geq (\alpha^* + \delta)\Delta$, where $\alpha^*\approx 1.763$ is the root of $x^x = \mathrm{e}$.
%A series of works~\cite{hayes2003randomly,hayes2003non, GMP05, HV06, molloy2004glauber, hayes2013local, dyer2013randomly} tried to relax the condition for rapid mixing. 

%Goldberg, Martin and Paterson~\cite{GMP05} proved that, for triangle-free and neighbourhood-amenable graphs, the Glauber dynamics has mixing time $O(n^2)$ if $q > \alpha^*\Delta$.   For more results about this problem, see \cite{frieze2007survey,chen2019improved};
%Hayes and Vigoda~\cite{HV06} proved $O(n \log n)$ mixing time
%Dyer~et.~al.~\cite{dyer2013randomly} proved $O(n \log n)$ mixing time if $q \geq (\alpha^* + \delta)\Delta$, girth $g \geq 5$ and $\Delta \geq \Delta_0(\delta)$; 
Compared with previous results, we achieved a $q \geq (\alpha^* + \delta)\Delta$ bound in triangle-free graphs without any additional requirements.
The condition in \Cref{theorem-colouring} matches the best known \emph{strong spatial mixing} regime for graph proper $q$-colourings~\cite{GMP05,GKM15}. 

\Cref{theorem-colouring} is proved via verifying the sufficient condition in \Cref{corollary-main}.
In fact, we apply the recursive coupling technique introduced by Goldberg, Martin and Paterson~\cite{GMP05} to bound the total influence caused by one vertex (the ``one-to-all infleunce''), namely to bound the induced $\infty$-norm of the influence matrix $\Psi^{\sigma_\Lambda}_\mu$, while the original approach in~\cite{GMP05} only provides bounds for ``one-to-one influence''.
%Our refinement is crucial since the original method in~\cite{GMP05} only provides bound for ``one-to-one influence'' and here we need bounds for ``one-to-all influence''.
%
Comparing to the traditional path coupling analysis, the power of the spectral independence approach lies in the fact that we can avoid considering the worst case scenario for the influence matrix in \Cref{def-inf-matrix}.
For path coupling, to avoid the worst case analysis one needs to establish so-called local uniformity \cite{hayes2013local},
which is difficult and causes various technical conditions in the results listed above.
In contrast, the method based on the spectral independence bypasses this obstacle.
%The verification can be done by going through the proofs of strong spatial mixing~\cite{GMP05,GKM15}.
%
%To obtain a tighter bound, we use the recursive coupling technique introduced by Goldberg, Martin and Paterson~\cite{GMP05} to bound the total influence caused by one vertex, namely, to bound induced $\infty$-norm of the influence matrix $\Psi^{\sigma_\Lambda}_\mu$.
%
%\Cref{theorem-colouring} then follows directly from \Cref{corollary-main}.
%%
%For our purpose, the parameters $C$ and $\eta$ in~\Cref{corollary-main} should depend only the constant gap $\delta$, thus we give a more careful analysis of recursive coupling based on self avoiding walks. 
%
%The techniques for analysing total influence should be of independent interest.
%

The downside of our result, similar to those of \cite{anari2020spectral,chen2020rapid}, is that the running time has a high exponent depending on how close the parameters are to the threshold.
Nonetheless, unlike the algorithm of \cite{liu2019deterministic}, our exponent remains a constant even if $\Delta=\omega(1)$, as long as we are below the threshold.

Finally, we remark that our refinement of recursive coupling argument might find applications in other problems. Armed with our notion of spectral independence, we essentially proved that the success of recursive coupling implies rapid mixing of Glauber dynamics for \emph{any} graph. This form of algorithmic implication was only known for special families of graphs like amenable graphs~\cite{GMP05} and planar graphs~\cite{yin2013approximate} before.

\section{Preliminaries}

\subsection{Linear algebra}
Let $v \in \mathbb{C}^n$ be an $n$-dimensional vector. For any integer $p \geq 1$, the $\ell_p$-norm of $v$ is defined by $\norm{v}_p = (\sum_{i=1}^n \abs{v_i}^p)^{1/p}$. Let $A \in \mathbb{C}^{n \times n}$ be a matrix.  For any integer $p \geq 1$, the induced $\ell_p$-norm of $A$ is defined by $\norm{A}_p = \sup_{v \in \mathbb{C}^n :\norm{v}_p = 1}\norm{Av}_p$. Let $\lambda_1,\lambda_2,\ldots \lambda_n \in \mathbb{C}$ be the eigenvalues of $A$. The \emph{spectral radius} of $A$ is defined by $\rho(A) \triangleq \max_{1 \leq i  \leq n}\abs{\lambda_i}$.	
The following relation is well-known.
\begin{proposition}[\text{\cite[Theorem 5.6.9. \& Corollary 5.6.14]{horn2012matrix}}]
\label{proposition-radius}
Let $A \in \mathbb{C}^{n \times n}$ be a matrix. 
For any integer $p \geq 1$, it holds that $\rho(A) \leq \norm{A}_p$ and $\lim_{k\rightarrow \infty}\Vert A^k \Vert^{1/k}_p = \rho(A)$.
%For any $\ell\geq 1$, $\lim_{k\rightarrow \infty}\Vert A^k \Vert^{1/k}_\ell = \rho(A)$.	
\end{proposition}

\subsection{Total variation distance and coupling}
Let $\mu$ and $\nu$ be two distributions over state space $\Omega$. The \emph{total variation distance} between $\mu$ and $\nu$ is defined by
\begin{align*}
\DTV{\mu}{\nu} \triangleq \frac{1}{2}\sum_{\*x \in \Omega}\abs{\mu(\*x) - \nu(\*x)}.	
\end{align*}
A coupling of $\mu$ and $\nu$ is a joint distribution $(X,Y) \in \Omega \times \Omega$ such that the marginal distribution of $X$ is $\mu$ and the marginal distribution of $Y$ is $\nu$. The following result is the well-known coupling inequality.

\begin{proposition}[\text{\cite[Proposition~4.7]{levin2017markov}}]
\label{proposition-coupling}
Let $\mu$ and $\nu$ be two distributions over state space $\Omega$. For any coupling $(X, Y)$ of $\mu$ and $\nu$, it holds that
\begin{align*}
\DTV{\mu}{\nu} \leq \Pr{X \neq Y}.
\end{align*}
Furthermore, there exists an \emph{optimal coupling} $(X, Y)$ such that $\DTV{\mu}{\nu} = \Pr{X \neq Y}$.
\end{proposition}

\subsection{Markov chain and mixing time}
Let $\Omega$ be a finite set which is the state space.
A Markov chain $(X_t)_{t\geq 0}$ on $\Omega$ is specified by transition matrix $P \in \mathbb{R}_{\geq 0}^{\Omega\times \Omega}$. 
We often identify the transition matrix with the corresponding Markov chain.
The Markov chain is \emph{irreducible} if for any $x, y \in \Omega$, there is a $t \geq 0$ such that $P^t(x,y) > 0$. 
The Markov chain is \emph{aperiodic} if for any $x \in \Omega$, $\gcd\{t > 0\mid P^t(x, x) > 0\} = 1$.
A distribution $\pi$ (viewed as a row vector) on $\Omega$ is \emph{stationary} with respect to a Markov chain $P$ if $\pi P = \pi$.
If a Markov chain $P$ is irreducible and aperiodic, then $P$ has a unique stationary distribution. 
A Markov chain is \emph{reversible} with respect to a distribution $\pi$ if the following \emph{detailed balance} condition holds
\begin{align}
\forall x,y \in \Omega,\quad \pi(x)P(x,y) = \pi(y)P(y,x),
\end{align}
which implies that $\pi$ is a stationary distribution of $P$.
All Markov chains considered in this paper are reversible.
In the following we state a few well-known spectral properties of reversible Markov chains.
%We give the proofs of these propositions for completeness.

\begin{proposition}[\text{\cite[Lemma~12.2]{levin2017markov}}]
\label{proposition-rever}
Let $\Omega$ be a finite set with $|\Omega| = n$. Let $\pi$ be a distribution with support $\Omega$. Let  $P \in \mathbb{R}_{\geq 0}^{\Omega \times \Omega}$ be the transition matrix of a Markov chain that is reversible with respect to $\pi$. Then
\begin{itemize}
\item $P$ has $n$ real eigenvalues $1 = \lambda_1 \geq \lambda_2 \geq \lambda_3 \geq \ldots \lambda_n \geq -1$;
\item there exist real eigenvectors $f_1,f_2,\ldots,f_n \in \mathbb{R}^\Omega$ such that $Pf_i = \lambda_if_i$ for all $1\leq i\leq n$, $f_1 = \vec{1}$ is a one-vector, and for any $1\leq i, j\leq n$,
\begin{align*}
	\sum_{x \in \Omega}f_i(x)f_j(x) \pi(x) = \*1[i = j].
\end{align*}
\end{itemize}
\end{proposition} 
%\begin{proof}
%Define a matrix $A$ such that $A(X,Y)=\sqrt{\frac{\pi(X)}{\pi(Y)}}P(X,Y)$. 
%Matrix $A$ is well-defined because $\pi(X) > 0$ for all $X \in \Omega$.
%Since $P$ is reversible with respect to $\pi$, $A$ is a real symmetric matrix.
%Therefore, $A$ has $n$ real eigenvalues $\lambda_1\leq\lambda_2\leq\ldots\leq\lambda_n$ and a real orthonormal basis $v_1,v_2,\ldots,v_n$ such that $Av_i = \lambda_iv_i$.
%Let $D$ be the diagonal matrix with $D(x,x) = \sqrt{\pi(x)}$. Then $A = DPD^{-1}$. Since $P$ is similar to $A$, then $P$ has the same real eigenvalues. Since $P$ is a stochastic matrix, $\norm{P}_{\infty} = 1$ and $\abs{\lambda_i} \leq \norm{P}_{\infty} = 1$. This proves the first part.
%
%Note that $\lambda_1 = 1$ and $v_1$ can be set as $v_1(x) = \sqrt{\pi(x)}$ for all $x \in \Omega$. Let $f_i = D^{-1}v_i$, then it holds that
%\begin{align*}
%Pf_i = P D^{-1}v_i = D^{-1}AD D^{-1}v_i = \lambda_i D^{-1} v_i = \lambda_i f_i.
%\end{align*}
%Thus, $f_1,f_2,\ldots,f_n$ are eigenvectors of $P$ and $f_1 = \vec{1}$.
%Since $v_1,v_2,\ldots,v_n$ form a real orthonormal basis,
%\begin{align*}
%\sum_{x \in \Omega}v_i(x)v_j(x) = \sum_{x \in \Omega}f_i(x)f_j(x)\pi(x) = \one{i=j}.	 
%\end{align*}
%This proves the second part.
%\end{proof}
We remark that \Cref{proposition-rever} holds if $P$ is reversible to $\pi$ and the support of $\pi$ is $\Omega$. It does not require $P$ to be irreducible. The following proposition bounds the mixing time of Markov chain.

\begin{proposition}[\text{\cite[Theorem~12.4]{levin2017markov}}]
\label{proposition-mixing}
Let $\Omega$ be a state space with $|\Omega| = n \geq 2$. Let $\pi$ be a distribution with support $\Omega$. Let  $P \in \mathbb{R}_{\geq 0}^{\Omega \times \Omega}$ be the transition matrix of a Markov chain that is reversible with respect to $\pi$. Let $1 = \lambda_1 \leq \lambda_2 \leq \ldots \lambda_n \leq -1$ be the real eigenvalues of $P$. Define the the \emph{absolute} spectral gap
\begin{align*}
 \gamma_{\star} \triangleq 1 - \lambda_{\star} = 1- \max\{\abs{\lambda_i} \mid  2\leq i \leq n\}.
\end{align*}
Let $\pi_{\min} \triangleq \min_{x\in \Omega}\pi(x)$. If $\gamma_{\star} > 0$, then it holds that
\begin{align*}
\forall\, 0< \epsilon < 1, \quad \tmix (\epsilon) \leq  \frac{1}{\gamma_{\star}} \left(  \log \frac{1}{\epsilon \pi_{\min}} \right),
\end{align*}
where $\tmix(\epsilon) \triangleq \max_{x \in \Omega} \min\{t \mid \DTV{P^t(x,\cdot)}{\pi} \leq \epsilon \}$ denotes the mixing time of Markov chain.
\end{proposition}
Note that the reversible chain $P$ is irreducible and aperiodic if the absolute spectral gap $ \gamma_{\star} > 0$. \Cref{proposition-mixing} says that $P$ converges to the unique stationary distribution $\pi$ rapidly if $\gamma_{\star} $ is bounded away from~$0$. See \cite[Theorem~12.4]{levin2017markov} for a formal proof of \Cref{proposition-mixing}.

We will use the following proposition to bound the absolute value of the second largest eigenvalue of $P$. Similar results appeared in \cite[Theorem~13.1]{levin2017markov} and \cite{MFC98}.
\begin{proposition}
\label{proposition-markov-mat}
Let $\Omega$ be a state space with $n = |\Omega| \geq 2$. Let $\pi$ be a distribution with support $\Omega$. Let  $P \in \mathbb{R}_{\geq 0}^{\Omega \times \Omega}$ be the transition matrix of a Markov chain that is reversible with respect to $\pi$. Then the second largest eigenvalue of $P$ satisfies
\begin{align*}
\forall t \geq 1,\quad  |\lambda_2|^t \leq d(t) \triangleq \max_{x,y\in \Omega}\DTV{P^t(x,\cdot)}{P^t(y,\cdot)}.
\end{align*}
\end{proposition}
\begin{proof}
Define a distance function $\delta$ on $\Omega$ as:
\begin{align*}
	\forall x,y \in \Omega:\quad \delta(x,y)\defeq \*1[x\ne y].
\end{align*}
For every function $f:\Omega \to \mathbb{R}$, define its Lipschitz constant with respect to $\delta$ as
\[
\-{Lip}(f)\defeq \max_{x,y\in \Omega:x\ne y} \frac{\abs{f(x)-f(y)}}{\delta(x,y)}.
\]
Fix a pair $x,y \in \Omega$, we use $\+C(x,y)$ to denote the optimal coupling between $P^t(x,\cdot)$ and $P^t(y,\cdot)$. 
Note that 
\begin{align*}
  P^tf(x) = \E[X\sim P^t(x,\cdot)]{f(X)}.
\end{align*}
Then for any $t \geq 1$, any function $f:\Omega \to \mathbb{R}$ and any $x,y \in \Omega$,
\begin{align*}
\abs{P^tf(x) - P^tf(y)} = \abs{\E[(X,Y)\sim \+C(x,y)]{f(X)-f(Y)}} \leq \E[(X,Y)\sim \+C(x,y)]{\abs{f(X)-f(Y)}},
\end{align*}
where the equality holds due to linearity of expectation. Then for any $t \geq 1$, any $f$ and any $x,y$,
\begin{align*}
\abs{P^tf(x) - P^tf(y)} \leq \-{Lip}(f)\Pr[(X,Y)\sim \+C(x,y)]{X\neq Y}	=\-{Lip}(f)\DTV{P^t(x,\cdot)}{P^t(y,\cdot)}   \leq \-{Lip}(f) d(t).
\end{align*}
Note that the inequality above holds for all $x, y \in \Omega$.
It implies that $\-{Lip}(P^tf) \leq \-{Lip}(f)d(t).$

Recall $|\Omega| = n$.
Let  $f_1,f_2,\ldots,f_n \in \mathbb{R}^\Omega$ be the eigenvectors in \Cref{proposition-rever}, where $f_1 = \vec{1}$.
Let $f=f_2$ be the eigenvector of $\lambda_2$, we have
\begin{align*}
	\abs{\lambda_2}^t\cdot\-{Lip}(f_2) = \-{Lip}(\lambda_2^t f_2)=\-{Lip}(P^tf_2)\le \-{Lip}(f_2)d(t).
\end{align*}
Note that $f_2 \neq \vec{0}$.
Since $f_1 = \vec{1}$ is a constant vector and $\sum_{x \in \Omega}f_1(x)\pi(x)f_2(x) = \sum_{x \in \Omega}\pi(x)f_2(x) = 0$, vector $f_2$ can not be a constant vector.
Thus, $\-{Lip}(f_2)>0$, we have $\abs{\lambda_2}^t\le d(t)$ for all $t\ge 1$.
\end{proof} 

One powerful technique to bound $\DTV{P^t(x,\cdot)}{P^t(y,\cdot)}$ is the \emph{coupling of Markov chain}.  A coupling of $P$ is a joint random process $(X_t,Y_t)_{t \geq 0}$ such that $(X_t)_{t \geq 0}$ and $(Y_t)_{t \geq 0}$ individually follow the transition rule of $P$, and if $X_k = Y_k$, then $X_t = Y_t$ for all $t \geq k$.
The following result follows from~\Cref{proposition-coupling}.
\begin{proposition}
\label{proposition-Markov-coupling}
Let $P$ be a Markov chain on state space $\Omega$ with a stationary distribution $\pi$. 
Let $X \in \Omega$ be a state.
Let $(X_t,Y_t)_{t \geq 0}$ be a coupling of Markov chain such that $X_0 = x_0$ and $Y_0=y_0$. Then
\begin{align*}
\forall t \ge 1, \quad
\DTV{P^t(x_0,\cdot)}{P^t(y_0,\cdot)} \leq \Pr{X_t \neq Y_t}.
\end{align*}
\end{proposition}

\section{Proof overview}
\label{section-overview}
In this section, we overview our proof of the main theorem (\Cref{theorem-main}).
We actually prove a slightly more general result. 
We first introduce the following definition of $(\eta_0,\eta_1,\ldots,\eta_{n-2})$-spectral independence, which is analogous to a similar notion in~\cite{anari2020spectral}.
\begin{definition}[$(\eta_0,\eta_1,\ldots,\eta_{n-2})$-Spectral Independence]
\label{def:spectral-independence-gen}
%We say a distribution $\mu$ over $[q]^V$, where $n=|V|$, 
%is \emph{$(\eta_0,\ldots,\eta_{n-2})$-spectrally independent}, 
%if for every $0\le k\le n-2$, $\Lambda\subseteq V$ of size $k$ and any feasible $\sigma_{\Lambda}\in[q]^\Lambda$, one has the spectral radius $\rho(\Psi^{\sigma_{\Lambda}}_{\mu})\le\eta_k$.
We say a distribution $\mu$ over $[q]^V$ with $|V|=n$ is \emph{$(\eta_0,\eta_1,\ldots,\eta_{n-2})$-spectrally independent}, 
if for every $0\le k\le n-2$, $\Lambda\subseteq V$ of size $k$ and any feasible $\sigma_{\Lambda}\in[q]^\Lambda$,
the spectral radius $\rho\left(\Psi^{\sigma_{\Lambda}}_{\mu}\right)$ of influence matrix $\Psi^{\sigma_{\Lambda}}_{\mu}$ satisfies
\begin{align*}
\rho\left(\Psi^{\sigma_{\Lambda}}_{\mu}\right) \leq \eta_k.
\end{align*}
\end{definition}

Since Glauber dynamics is reversible with respect to $\mu$, its transition matrix has real eigenvalues.
The following theorem gives a lower bound on its spectral gap when $\mu$ is spectrally independent.

\begin{theorem}
\label{theorem-general}
Let $\mu$ be a distribution over $[q]^V$, where $n=|V|$.
Let $\eta_0,\eta_1,\ldots,\eta_{n-2}$ be a sequence where $0\leq \eta_k <n-k-1$ for all $0\leq k \leq n-2$.
If $\mu$ is $(\eta_0,\eta_1,\ldots,\eta_{n-2})$-spectrally independent, % for a sequence $\eta_0,\eta_1,\ldots,\eta_{n-2}$ such that $0\leq \eta_k <n-k-1$ for all $0\leq k \leq n-2$, 
then the Glauber dynamics for $\mu$ has spectral gap
\begin{align*}
1 - \lambda_2(P_{\mathsf{Glauber}}) \geq  \frac{1}{n}\prod_{k=0}^{n-2}\tp{1 - \frac{\eta_k}{n-k-1}},
\end{align*}
where %$\tmix(\epsilon) = \max_{X_0 \in \Omega} \min\left\{t \mid \DTV{P_{\mathsf{Glauber}}^t(X_0,\cdot)}{\mu} \leq \epsilon \right\}$ and 
$\lambda_2(P_{\mathsf{Glauber}})$ is the second largest eigenvalue of transition matrix $P_{\mathsf{Glauber}}$. 
%and
%$\mu_{\min}\triangleq \min\{\mu(\sigma)\mid\sigma\in[q]^V\land \mu(\sigma)>0\}$.	
\end{theorem}

We prove this theorem on general domain of size $q\ge 2$.
The main theorem (\Cref{theorem-main}) is a corollary of \Cref{theorem-general}, because if $\mu$ is $(C,\eta)$-spectrally independent, then $\mu$ is $(\eta_0,\eta_1,\ldots,\eta_{n-2})$-spectrally independent with $\eta_k = \min\left\{ C, \eta(n-k-1) \right\}$.
We first give a proof overview of \Cref{theorem-general}, then prove main theorem (\Cref{theorem-main}) via \Cref{theorem-general} in~\Cref{section-proof-main}.

\subsection{Glauber dynamics and local random walks}
%Given a distribution $\mu$ over $[q]^V$, let $P_{\mathsf{Glauber}}$ denote the transition matrix of Glauber dynamics for $\mu$.
To prove \Cref{theorem-general},
we first interpret the Glauber dynamics on $[q]^V$ as a down-up random walk on simplicial complexes.
Then we apply the local-to-global theorem due to Alev and Lau~\cite{alev2020improved} to reduce the task of analysing Glauber dynamics (a global random walk) to the task of analysing local random walks. 
Similar routines have been applied in several previous works~\cite{ALOV19, CGM19, anari2020spectral,chen2020rapid}.
\begin{definition}[Local Random Walk]
\label{definition-local-walk}
For any subset $\Lambda \subseteq V$, any feasible partial configuration $\sigma_{\Lambda} \in [q]^{\Lambda}$,  
define local random walk $P_{\sigma_\Lambda}$ on $U_{\sigma_\Lambda}=\{(u,c) \in \overline{\Lambda} \times [q] \mid \mu_u^{\sigma_\Lambda}(c) > 0\}$ as
\begin{align}
\label{eq-def-one-skeleton-walk-spin}
\forall (u,i),(v,j) \in U_{\sigma_\Lambda},\quad P_{\sigma_\Lambda}((u,i), (v,j)) \triangleq  \frac{\one{u\neq v}}{\abs{V} - \abs{\Lambda}-1} \mu^{\sigma_{\Lambda}, u \gets i}_v(j),
\end{align}
where $\overline{\Lambda}=V \setminus \Lambda$, and  $ \mu^{\sigma_{\Lambda}, u \gets i}_v$ is the marginal distribution on $v$ induced from $\mu$ conditional on the configuration on $\Lambda$ fixed as $\sigma_\Lambda$ and that $u$ is fixed to $i$.
\end{definition}

%The state space of local random walk $P_{\sigma}$ is polynomial in $n$, while the state space of Glauber dynamics is exponentially large. 
%The local-to-global technique~\cite{alev2020improved} 
\Cref{corollary-AL} below shows that the second largest eigenvalue $\lambda_2(P_{\mathsf{Glauber}})$ of Glauber dynamics is small as long as the second largest eigenvalues $\lambda_2(P_{\sigma_\Lambda})$\footnote{Local random walk $P_{\sigma_\Lambda}$ has real eigenvalues because $P_{\sigma_\Lambda}$ is reversible. See \Cref{section-reduction} for more details.} of local random walks are all small. %Formally, we have the following result.
\begin{condition}
\label{condition-local-expander}
Let $\mu$ be a distribution over $[q]^V$, where $n=|V|$.
There exists a sequence $\alpha_0,\alpha_1,\ldots,\alpha_{n-2}$ such that for every $0\le k\le n-2$, $\Lambda\subseteq V$ of size $k$ and any feasible $\sigma_{\Lambda}\in[q]^\Lambda$, the transition matrix $P_{\sigma_\Lambda}$ satisfies
\begin{align*}
\lambda_2(P_{\sigma_\Lambda}) \leq \alpha_{k},	
\end{align*}
where $\lambda_2(P_{\sigma_\Lambda})$ is the second largest eigenvalue of the matrix $P_{\sigma_\Lambda}$.
\end{condition}
%Combining~\Cref{observation-spin-complexes} and \Cref{theorem-AL}, we have the following result.
\begin{lemma}[\cite{alev2020improved}]
\label{corollary-AL}
Let $\mu$ be a distribution over $[q]^V$, where $n=|V|$.
Let $\alpha_0,\alpha_1,\ldots,\alpha_{n-2}$ be a sequence where $0\leq \alpha_i <1$ for all $0\leq i \leq n-2$.
%Let $P_{\mathsf{Glauber}} $ denote the Glauber dynamics for $\mu$.
If $\mu$ satisfies~\Cref{condition-local-expander} with $\alpha_0,\alpha_1,\ldots,\alpha_{n-2}$, then the Glauber dynamics for $\mu$ has spectral gap
\begin{align*}
1 - \lambda_2(P_{\mathsf{Glauber}}) \geq \frac{1}{n}\prod_{k = 0}^{n-2}(1 - \alpha_k),
\end{align*}
where $\lambda_2(P_{\mathsf{Glauber}})$ is the second largest eigenvalue of transition matrix  $P_{\mathsf{Glauber}}$.
\end{lemma}

\Cref{corollary-AL} relates Glauber dynamics to local random walks, which provides a powerful tool to analyse Glauber dynamics, 
because the state space of local random walks are exponentially smaller compared to that of Glauber dynamics.
 %\Cref{corollary-AL} is proved in \Cref{section-reduction}. 
\Cref{corollary-AL} (proved in \Cref{section-reduction}) is an easy corollary of the main result in~\cite{alev2020improved}.
%We will prove \Cref{corollary-AL} in \Cref{section-reduction}.

\subsection{Analysis of local random walks}
Our remaining task is to bound the second largest eigenvalues of local random walks. 

Our main technical contribution is the following lemma (proved in \Cref{section-coupling}), which states that for distribution $\mu$ over $[q]^V$ with general domain size $q\ge 2$, these second largest eigenvalues of local random walks are always  small if $\mu$ is spectrally independent.

%Formally, we prove the following lemma in \Cref{section-coupling}.
\begin{lemma}
\label{lemma-decay-imples-lambda2}
Let $\mu$ be a distribution over $[q]^V$, where $n=|V|$.
%Let $\eta_0,\eta_1,\ldots,\eta_{n-2}$ be a sequence where $0\leq \eta_i <1$ for all $0\leq i \leq n-2$.
If $\mu$ is $(\eta_0,\eta_1,\ldots,\eta_{n-2})$-spectrally independent, % for $\eta_0,\eta_1,\ldots,\eta_{n-2}\in[0,1)$,
then $\mu$ satisfies \Cref{condition-local-expander} with a sequence $\alpha_0,\alpha_1,\ldots,\alpha_{n-2}$ such that %for all $0\leq k \leq n-2$,
\begin{align*}
\forall 0\le k\le n-2:\quad \alpha_k = \frac{\eta_k}{n-k-1}.	
\end{align*}
\end{lemma}
For the special case $q = 2$, Anari, Liu and Oveis Gharan~\cite{anari2020spectral} proved a similar version of \Cref{lemma-decay-imples-lambda2}. 
They used a linear algebra argument %to relate signed influence matrices to the transition matrices of local random walks. 
to identify the second largest eigenvalue of the local random walk with the largest eigenvalue of the signed influence matrix.
In such analysis, some key identities crucially rely on that $q = 2$, which makes it hard to extend it to general domains of size $q>2$.
%Their definition of influence matrix and key identity~\cite[Claim~3.2]{anari2020spectral} of the proof both relies on the fact of $q = 2$.

Alternatively, we propose a new \emph{coupling} based argument to show the rapid mixing of the local random walk $P_{\sigma_\Lambda}$, assuming the spectral independence, which implies an upper bound of $\lambda_2(P_{\sigma_\Lambda})$.
Specifically, we construct a coupling $(X_t,Y_t)_{t\geq 0}$ for each local random walk, and show that the two chains coalesce (namely $X_t = Y_t$) quickly if $\mu$ is spectrally independent.
Then we use \Cref{proposition-markov-mat} and \Cref{proposition-Markov-coupling} to bound the second largest eigenvalue.
Our coupling argument is simple and combinatorial, reminiscent of an analysis by Hayes \cite{hayes2006simple}.
And it has the advantage of being applicable to joint distributions with general domain sizes.
Note that here we are only giving an upper bound for $\lambda_2(P_{\sigma_\Lambda})$, which is sufficient for our purpose, rather than establishing the equality as in \cite{anari2020spectral} for the case with $q=2$.
%
%Nonetheless this direction is what we need to show rapid mixing.
%
Detailed analysis is given in \Cref{section-coupling}.

\subsection{Proof of main theorem}
\label{section-proof-main}
It is straightforward to verify that \Cref{theorem-general} is a corollary of \Cref{corollary-AL} and \Cref{lemma-decay-imples-lambda2}.
We now use \Cref{theorem-general} to prove the main theorem (\Cref{theorem-main}).
\begin{proof}[Proof of \Cref{theorem-main}]
Since $\mu$ is $(C,\eta)$-spectrally independent (\Cref{def:spectral-independence}) for $C \geq 0$ and $1\leq \eta < 1$, by \Cref{def:spectral-independence-gen},  $\mu$ is $(\eta_0,\eta_1,\ldots,\eta_{n-2})$-spectrally independent for 
\begin{align*}
\eta_k = \min\left\{ C, \eta(n-k-1) \right\}.	
\end{align*}
By \Cref{theorem-general}, we have
\begin{align*}
1 - \lambda_2(P_{\mathsf{Glauber}}) 
&\geq \frac{1}{n}\prod_{k = 0}^{n-2}\tp{1 - \frac{\eta_k}{n-k-1}} \geq \frac{1}{n}\prod_{k = 0}^{n-2}\tp{1-\min\left\{\frac{C}{n-k-1}, \eta  \right\}}=\frac{1}{n}\prod_{k = 1}^{n-1}\tp{1-\min\left\{\frac{C}{k}, \eta  \right\}}.
\end{align*}
Thus, the spectral gap has the following lower bound
\begin{align*}
1 - \lambda_2(P_{\mathsf{Glauber}})&\geq \frac{1}{n}\left( \prod_{k=1}^{2+2C-1}(1-\eta) \right) 	\left( \prod_{k = 2+2C}^{n-1}\left(1-\frac{C}{k}\right) \right)\geq \frac{(1-\eta)^{2+2C}}{n} \prod_{k = 2+2C}^ n\left(1-\frac{C}{k}\right)\\
 &\geq \frac{(1-\eta)^{2+2C}}{n}\exp\left(-\sum_{k=2+2C}^{n}\frac{2C}{k}\right) \geq \frac{(1-\eta)^{2+2C}}{n}\exp\left(-2C\sum_{k=2}^n\frac{1}{k}\right)\\
 (\star)\quad &\geq \frac{(1-\eta)^{2+2C}}{n}\exp\left(-2C\ln n\right) = \frac{(1-\eta)^{2+2C}}{n^{1+2C}},
 \end{align*}
where $(\star)$ holds because $\sum_{k = 2}^n\frac{1}{k} \leq  \ln n$.

Since the transition matrix of Glauber dynamics is positive semi-definite,  
all of its eigenvalues are real~\cite{DGU14,alev2020improved}.
Let the eigenvalues be $1 = \lambda_1(P_{\mathsf{Glauber}}) \geq \lambda_2(P_{\mathsf{Glauber}}) \geq \ldots \geq \lambda_{\abs{\Omega}}(P_{\mathsf{Glauber}}) \geq 0$, where $\Omega \subseteq [q]^V$ is the support of $\mu$. 
The absolute spectral gap of Glauber dynamics has the following lower bound
\begin{align*}
\gamma_{\star} = 1 - \lambda_{\star} = 1 - \max_{2\leq i \leq \abs{\Omega}}\abs{\lambda_i(P_{\mathsf{Glauber}})} = 1 - \lambda_2(P_{\mathsf{Glauber}}) \geq \frac{(1-\eta)^{2+2C}}{n^{1+2C}}.
\end{align*}
By \Cref{proposition-mixing}, we have
\begin{align*}
\tmix (\epsilon) &\leq  \frac{1}{\gamma_{\star}} \left(  \log \frac{1}{\epsilon \mu_{\min}} \right)  
\leq   \frac{n^{1+2C}}{(1-\eta)^{2+2C}}\left(  \log \frac{1}{\epsilon \mu_{\min}} \right).\qedhere
\end{align*}
%This implies $\tmix(\eps) \leq \frac{n^{1+2C}}{(1-\eta)^{2+2C}}\left(  \log \frac{1}{\epsilon \mu_{\min}} \right)$.
%\Cref{theorem-main} follows from the well-known relation between the mixing time and the spectral gap~\cite[Theorem 12.4]{levin2017markov}.  	
\end{proof}

\section{Simplicial complexes and Glauber dynamics}
\label{section-reduction}
In this section, we relate the Glauber dynamics to the random walk on simplicial complexes.
As explained in the proof overview, we reduce the task of giving an upper bound for the second largest eigenvalue of Glauber dynamics 
to the task of giving an upper bound for the second largest eigenvalues of local random walks (\Cref{corollary-AL}).
\subsection{Simplicial complexes and random walks}
Let $U$ be a ground set. 
A \emph{simplicial complex} $X \subseteq 2^U$ is a family of subset that is downward closed, i.e.~if $\alpha \in X$, then $\beta \in X$ for all $\beta \subseteq \alpha$. 
Each subset $\alpha \in X$ is called a \emph{face}. 
The \emph{dimension} of a face $\alpha$ is its size $|\alpha|$.\footnote{In some papers, such as \cite{KO20,alev2020improved}, the dimension is defined to be $|\alpha|-1$.}
We use $X(j)$ to denote the set of faces with dimension $j$. 
The dimension of a simplicial complex $X$ is the maximum dimension of all its faces.
We call $X$ a \emph{pure} $d$-dimensional simplicial complex if every maximal face of $X$ is of dimension $d$. We only consider pure simplicial complexes in this paper.

We consider the weighted simplicial complexes.
Let $X$ be a pure $d$-dimensional simplicial complex. Given a weight function $\Pi: X(d) \rightarrow \mathbb{R}_{\geq 0}$, define the induced weights for all faces in $X$ by
\begin{align}
\label{eq-induced-weight}
\forall \alpha \in X, \quad \Pi(\alpha) = \sum_{\beta \in X(d): \beta \supseteq \alpha} \Pi(\beta).	
\end{align}

For each face $\alpha \in X$, the \emph{link} $X_{\alpha}$ is simplicial complexes defined by
\begin{align*}
X_{\alpha} \triangleq \{\beta \setminus \alpha \mid \beta \in X \land \alpha \subseteq \beta \}.
\end{align*}
Let $\Pi_{\alpha}$ be the weight of $X_{\alpha}$ induced from $\Pi$, i.e. for each face $\beta \in X_{\alpha}$,
\begin{align*}
\Pi_{\alpha}(\beta) \triangleq \Pi(\alpha \uplus \beta).	
\end{align*}

The \emph{one-skeleton} of $X_{\alpha}$ is a weighted graph $G_{\alpha} = (V_\alpha, E_\alpha, \Phi_\alpha)$, where $V_\alpha = X_{\alpha}(1)$ is the set of singletons,  $E_{\alpha}=X_{\alpha}(2)$ is the set of 2-dimensional faces, and $\Phi_{\alpha}(u,v) = \Pi_{\alpha}(\{u,v\})$ for all $\{u,v\} \in E_{\alpha}$.
We use $P_{\alpha}$ to denote the simple (non-lazy) random walk on one-skeleton $G_{\alpha}$. The transition probability is defined by
\begin{align}
\label{eq-def-one-skeleton-walk}
\forall u, v \in V_{\alpha},\quad P_{\alpha}(u, v) \triangleq 	\begin{cases}
 \frac{\Phi_\alpha(u,v)}{\sum_{w:\{u,w\} \in E_{\alpha}}\Phi_{\alpha}(u,w)}		& \text{if } \{u,v\} \in E_{\alpha};\\
 0		& \text{if } \{u,v\} \notin E_{\alpha}.
 		 	\end{cases}
\end{align}

Given a pure $d$-dimensional weighted simplicial complexes $(X, \Pi)$, define the following down-up random walk $P_d^{\lor}$ on $X(d)$. Suppose the current state is $\sigma_t \in X(d)$, the next state $\sigma_{t+1} \in X(d)$ is generated as follows
\begin{itemize}
\item (down walk) pick  $x \in \sigma_t$ uniformly at random, and drop $x$ to obtain $\sigma' = \sigma_t \setminus \{x\} \in X(d-1)$;
\item (up walk) sample $\sigma_{t+1} \in X(d)$ satisfying $\sigma' \subseteq \sigma_{t+1}$ with probability proportional to $\Pi(\sigma_{t+1})$.
\end{itemize}
Therefore, the transition matrix of down-up random walk is defined by
\begin{align*}
\forall \alpha,\beta \in X(d), \quad P_{d}^{\lor}(\alpha, \beta) \triangleq \begin{cases}
\sum_{\tau \in X(d-1):\tau \subset \alpha}\frac{\Pi(\alpha)}{d \cdot \Pi(\tau)} & \text{if } \alpha = \beta;\\ 
\frac{\Pi(\beta)}{d\cdot \Pi(\alpha \cap \beta)} & \text{if } \alpha \cap \beta \in X(d-1);\\
0 & \text{otherwise}.
 \end{cases}
\end{align*}

The relation between down-up random walk $P^\lor_d$ and random walks one-skeletons $P_{\alpha}$ was studied in many works~\cite{Opp18,KO20,alev2020improved}.
Note that both random walks $P^\lor_d$ and $P_{\alpha}$ are reversible. By \Cref{proposition-rever}, both of them have real eigenvalues.
\begin{definition}[local-spectral expander~\cite{Opp18,KO20,alev2020improved}]
\label{definition-local-expander}
Let $(X,\Pi)$ be a pure $d$-dimensional weighted simplicial complexes.
We say that $(X,\Pi)$ is a $(\gamma_0,\gamma_1,\ldots,\gamma_{d-2})$-local-spectral expander if for any $0\leq k \leq d-2$, it holds that
\begin{align*}
\max\{\lambda_2(P_{\alpha}) \mid  \alpha \in X(k) \}	 \leq \gamma_k,
\end{align*}
where $\lambda_2(P_{\alpha})$ stands for the second largest eigenvalue of $P_{\alpha}$, and $P_{\alpha}$, as defined in~\eqref{eq-def-one-skeleton-walk}, is the transition matrix for the simple (non-lazy) random walk on one-skeleton of link $X_{\alpha}$.
\end{definition}

\begin{theorem}[\cite{alev2020improved}]
\label{theorem-AL}
Let $(X,\Pi)$ be a pure $d$-dimensional weighted simplicial complexes. If $(X,\Pi)$ is 	a $(\gamma_0,\gamma_1,\ldots,\gamma_{d-2})$-local-spectral expander, then 
\begin{align*}
 \lambda_2(P_d^\lor) \leq  1- \frac{1}{d}\prod_{k=0}^{d-2}(1-\gamma_k),
\end{align*}
where $\lambda_2(P_d^\lor)$ is the second largest eigenvalue of down-up random walk $P^{\lor}_d$.
\end{theorem}

We remark that the chain $P^{\lor}_d$ is denoted as $P^{\triangledown}_{d-1}$ in~\cite{alev2020improved}.

\subsection{Connections to Glauber dynamics}
%Let $\+I=(G,[q],\*b,\*A)$ be a spin system with Gibbs distribution $\mu$, where $G=(V,E)$ and $|V| = n$. 
Let $\mu$ be a distribution over $[q]^V$, where $|V| = n$.
Let $\Omega \subseteq [q]^V$ be the support of $\mu$.
We define a ground set of $nq$ elements
\begin{align*}
U \triangleq \{(u,c) \mid u \in V \land c \in [q] \}.	
\end{align*}
For each (possibly partial) configuration $\sigma\in [q]^\Lambda$ where $\Lambda\subseteq V$, we associate with it a face $f_{\sigma} \subseteq U$ as
\begin{align*}
f_{\sigma} \triangleq \{(u,\sigma_u) \mid u \in \Lambda\}.	
\end{align*}
Let $X$ be the downward closure of the family of faces $\{f_\sigma \mid \sigma \in \Omega\}$. 
Then $X$ is a pure $n$-dimensional simplicial complex. For each maximal face $f_\sigma \in X$ where $\sigma \in \Omega$, we assign a weight according to $\mu$
\begin{align*}
\Pi(f_\sigma) = \mu(\sigma),
\end{align*}
and each face in $X$ obtains an induced weight from~\eqref{eq-induced-weight}. Hence, $(X,\Pi)$ is a weighted pure $n$-dimensional simplicial complex.
The following observation is straightforward to verify.
One way to understand it is to view the state space $[q]^V$ as the set of bases of a partition matroid.
\begin{observation}
\label{observation-spin-complexes}
The Glauber dynamics on $\mu$ is precisely the down-up random walk on $X(n)$.	
\end{observation}

Let $\Lambda \subseteq V$ be a subset of variables. 
For every feasible partial configuration $\sigma =\sigma_\Lambda \in [q]^{\Lambda}$, there exists a face $f_{\sigma} = \{(u,\sigma_u)\mid u \in \Lambda\}$ in $X$, and vice versa. 

To simplify the notation, we use $P_{\sigma}$ to denote the simple (non-lazy) random walk $P_{f_{\sigma}}$ (defined in~\eqref{eq-def-one-skeleton-walk}) on one-skeleton of link $X_{f_{\sigma}}$. By definition, $P_{\sigma}$ is a random walk on $U_{\sigma}=\{(u,c) \in \overline{\Lambda} \times [q] \mid \mu_u^{\sigma}(c) > 0\}$, where $\overline{\Lambda} = V \setminus \Lambda$. 
Fix $\*x = (u,i) \in U_\sigma$ and $\*y = (v,j) \in U_\sigma$. 
The weight of  edge $\{\*x,\*y\}$ in one-skeleton of link $X_{f_{\sigma}}$ is given by
\begin{align*}
\Phi_{f_\sigma}(\*x,\*y) = 
\sum_{\substack{\tau \in \Omega\\\tau_\Lambda=\sigma,\tau_u=i,\tau_v=j}}\mu(\tau) = \Pr[X \sim \mu]{X_\Lambda = \sigma \land X_u =i \land X_v =j}. 
\end{align*}
Thus,
\begin{align*}
\forall (u,i),(v,j) \in U_{\sigma},\quad P_{\sigma}((u,i), (v,j)) =  \frac{\one{u\neq v}}{\abs{V}-\abs{\Lambda}-1} \mu^{\sigma, u \gets i}_v(j),
\end{align*}
where $ \mu^{\sigma, u \gets i}_v$ is the marginal distribution on $v$ induced from $\mu$ conditional on the configurations on $\Lambda$ fixed as $\sigma$ and that $u$ is fixed to $i$. This is precisely the local random walk in \Cref{definition-local-walk}.
Note that $P_{\sigma}$ is reversible because the random walk on one-skeleton is reversible.
%in particular, if $u=v$, the distribution $\mu_u^{\overline{\Lambda} \gets \sigma, u \gets i}(j) \triangleq \one{i=j}$ for all $j \in [q]$.

%We introduce the following condition for spin system $\+I$, which is analogous to \Cref{definition-local-expander}.
Note that if $\mu$ satisfies \Cref{condition-local-expander} with $\alpha_0,\alpha_1,\ldots,\alpha_{n-2}$ , then the $n$-dimensional weighted simplicial complex $(X,\Pi)$ defined above is a $(\gamma_0,\gamma_1,\ldots,\gamma_{n-2})$-local-spectral expander with $\gamma_k = \alpha_k$ for all $0\leq k \leq n- 2$.
Hence, \Cref{corollary-AL} is a corollary of \Cref{theorem-AL}.

%\pagebreak
\section{Analysis of local random walks}
\label{section-coupling}
%In this section, we prove the main theorem (\Cref{theorem-main}). We need the following lemma.
%\begin{lemma}
%\label{lemma-decay-imples-lambda2}
%Let $\+I=(G,[q],\*b,\*A)$ be a spin system, where $G = (V, E)$ is an $n$-vertex undirected graph.
%If $\+I$ satisfies \Cref{condition-main} with constants $C>0$ and $0<\eta <1$, then $\+I$ satisfies \Cref{condition-local-expander} with a sequence $\alpha_2,\alpha_3,\ldots,\alpha_n$ such that for all $1\leq i \leq n$,
%\begin{align*}
%\alpha_i \leq \min\left\{\frac{C}{i-1}, \eta\right\}.	
%\end{align*}
%\end{lemma}
In this section, we prove \Cref{lemma-decay-imples-lambda2}, which states that the second largest eigenvalues of local random walks are  small if the distribution $\mu$ is spectrally independent. 
The new ingredient of this part is a coupling based argument for the above implication for distributions with general domain size $q$.

\subsection{Proof of \texorpdfstring{\Cref{lemma-decay-imples-lambda2}}{4.1}}
Fix a subset $\Lambda \subseteq V$ with $0\leq |\Lambda|\leq n-2$, and a feasible partial configuration $\sigma_\Lambda \in [q]^{\Lambda}$. To simplify the notation, we use $\sigma$ to denote $\sigma_{\Lambda}$. Consider the random walk $P_\sigma$ defined in~\eqref{eq-def-one-skeleton-walk-spin}. Recall the state space of $P_{\sigma}$ is defined by
\begin{align}
\label{eq-def-U}
U_{\sigma} \triangleq \left\{(u,i) \in \overline{\Lambda} \times [q] \mid \mu^{\sigma}_u(i) > 0 \right\}.
\end{align}
Thus $\abs{U_\sigma} \geq 2$ because $|\overline{\Lambda}| = \abs{V \setminus \Lambda} \geq 2$ and $\sigma$ is feasible.
Consider a random walk $Q_{\sigma}$:
\begin{align}
\label{eq-def-Q}
Q_\sigma \defeq \frac{n -|\Lambda| - 1}{n - |\Lambda|}P_{\sigma} + \frac{1}{n - |\Lambda|} I_{\sigma},
\end{align}
where $I_\sigma \in \mathbb{R}_{\geq 0}^{U_{\sigma} \times U_{\sigma}}$ is the identity matrix. 
In other words, in each step, with probability $\frac{1}{n-\abs{\Lambda}}$, the random walk $Q_{\sigma}$ stays at the current state;
otherwise, $Q_{\sigma}$ evolves in the same way as $P_{\sigma}$. 

Define a distribution $\pi$ over $U_{\sigma}$ as 
\begin{align}
\label{eq-def-pi}
\forall (u,i) \in U_{\sigma},\quad \pi(u,i) \triangleq \frac{1}{n - \abs{\Lambda}}\mu^\sigma_u(i).
\end{align}
Note that $\sum_{i\in \Omega^\sigma_u}\mu^\sigma_u(i)=1$,
where $\Omega^\sigma_u \triangleq \{i \in [q] \mid \mu^\sigma_u(i) > 0\}$.
Thus $\sum_{(u,i)\in U_{\sigma}}\pi(u,i) = 1$ and $\pi$ is well-defined.
We claim that both $P_{\sigma}$ and $Q_{\sigma}$ are reversible with respect to $\pi$. 
For any $(u,i),(v,j) \in U_{\sigma}$, we verify the detailed balance equation.
If $u = v$, then it is straightforward to verify
\begin{align*}
\pi(u,i)P_{\sigma}((u,i),(v,j))  = 0 = \pi(v,j) P_{\sigma}((v,j),(u,i));	
\end{align*}
otherwise  $u \neq v$, then
\begin{align*}
\pi(u,i)P_{\sigma}((u,i),(v,j))	&= \frac{\mu^\sigma_u(i)\cdot \mu^{\sigma, u \gets i}_v(j) }{(n-\abs{\Lambda})(n-\abs{\Lambda}-1)} = \frac{\Pr[X \sim \mu]{X_u=i \land X_v = j \mid X_{\Lambda}=\sigma}}{(n-\abs{\Lambda})(n-\abs{\Lambda}-1)}\\
&=\frac{\mu^\sigma_v(j) \cdot \mu^{\sigma, v \gets j}_u(i) }{(n-\abs{\Lambda})(n-\abs{\Lambda}-1)} = \pi(v,j)P_{\sigma}((v,j),(u,i)).
\end{align*}
Since $Q_{\sigma}$ is a lazy version of $P_{\sigma}$, $Q_{\sigma}$ is also reversible to $\pi$.
By~\eqref{eq-def-U} and~\eqref{eq-def-pi}, the support of $\pi$ is $U_{\sigma}$. 
By \Cref{proposition-rever}, $P_{\sigma}$ and $Q_{\sigma}$ both have $\abs{U_\sigma}$ real eigenvalues. Let $\lambda_2(P_{\sigma})$ and $\lambda_2(Q_{\sigma})$ denote the second largest eigenvalues of $P_{\sigma}$ and $Q_{\sigma}$.
By the definition of $Q_\sigma$ in~\eqref{eq-def-Q}, we have the following proposition.
\begin{proposition}
\label{proposition-lambda-relation}
$\lambda_2(Q_\sigma) =  \frac{n -|\Lambda| - 1}{n - |\Lambda|}\lambda_2(P_{\sigma}) + \frac{1}{n-|\Lambda|}$.
\end{proposition}
\Cref{proposition-lambda-relation} is a basic result in linear algebra.
We claim the following result about $\lambda_2(Q_{\sigma})$.
\begin{lemma}
\label{lemma-lambda2-Q}	
%If $\+I$ satisfies \Cref{condition-main} with constants $C>0$ and $0<\eta <1$, then
%\begin{align*}
$\lambda_2(Q_\sigma) \leq \frac{\rho(\Psi^\sigma_\mu)+1}{n-\abs{\Lambda}}.$
%\end{align*}
\end{lemma}
The proof of \Cref{lemma-lambda2-Q} is deferred to the next subsection. We now use \Cref{proposition-lambda-relation} and \Cref{lemma-lambda2-Q} to prove \Cref{lemma-decay-imples-lambda2}. 
Suppose $\mu$ is $(\eta_0,\eta_1,\ldots,\eta_{n-2})$-spectrally independent (\Cref{def:spectral-independence-gen}).
By \Cref{proposition-lambda-relation}, it holds that
\begin{align*}
\lambda_2(P_{\sigma}) = \frac{n -\abs{\Lambda}}{n - \abs{\Lambda}-1}\left(\lambda_2(Q_\sigma)-\frac{1}{n - \abs{\Lambda}}\right) \leq \frac{\rho(\Psi^\sigma_\mu)}{n -\abs{\Lambda}-1} \leq \frac{\eta_k}{n-k-1}, \text{ where } k = \abs{\Lambda}.
\end{align*}
The above inequality holds for any $\Lambda \subseteq V$ with $0\leq \abs{\Lambda} \leq n-2$ and any feasible $\sigma \in [q]^{\Lambda}$. This implies $\mu$ satisfies \Cref{condition-local-expander} with $\alpha_0,\alpha_1,\ldots,\alpha_{n-2}$ such that $\alpha_k = \frac{\eta_k}{n-k-1}$.

\subsection{A coupling based analysis} We now prove \Cref{lemma-lambda2-Q}. 
We will use coupling to give an upper bound of $\lambda_2(Q_{\sigma})$.
First we define a matrix $A$:
\begin{align}
\label{eq-def-A}
A \triangleq \frac{1}{n-\abs{\Lambda}} \tp{\tp{\Psi_{\mu}^\sigma}^T + I},
\end{align}
where $I$ is the identity matrix.
For any $t \geq 1$, define 
\begin{align}
\label{eq-def-dt}
d(t) \triangleq \max_{x_0,y_0 \in U_{\sigma}}\DTV{Q_\sigma^t(x_0,\cdot)}{Q_\sigma^t(y_0,\cdot)}.
\end{align}
%We first prove the following result.
\begin{lemma}
\label{lemma-coupling}
For any $t \geq 1$,
$d(t) \leq \norm{A^{t-1}}_1$. 	
\end{lemma}
\begin{proof}
By the definitions of $Q_\sigma$ in~\eqref{eq-def-Q} and $P_{\sigma}$ in~\eqref{eq-def-one-skeleton-walk-spin}, we have
\begin{align*}
\forall (u,i),(v,j) \in U_\sigma,\quad
Q_{\sigma}((u,i), (v,j)) = \frac{\mu^{\sigma, u \gets i}_v(j)}{n-\abs{\Lambda}} ,
\end{align*}
where if $u=v$, the distribution $\mu_u^{\sigma, u \gets i}(j) = \one{i=j}$ for all $j \in [q]$.
Let $X_0,X_1,X_2,\ldots \in U_{\sigma}$ be the sequence of random states generated by $Q_{\sigma}$, where $X_t = (\Xvtx_t , \Xcol_t)$, $\Xvtx_t \in V \setminus \Lambda$ and $\Xcol_t \in [q]$. 
By definition of $Q_\sigma$ in~\eqref{eq-def-Q},
given $X_{t-1} = (u, i)$, the random pair $X_{t} = (v, j)$ can be generated by the following procedure
\begin{itemize}
\item sample $v \in V \setminus \Lambda$ uniformly at random;
\item sample $j \in [q]$ from the distribution $\mu^{ \sigma, u \gets i}_v(\cdot)$.
\end{itemize}

Next, we define a coupling procedure $\+C$.
Let $(X_t)_{t \geq 0}$ be the random walk $Q_{\sigma}$ starting from $X_0 = \*x_0 \in U_{\sigma}$, and $(Y_t)_{t \geq 0}$ be the random walk $Q_{\sigma}$ starting from  $Y_0 = \*y_0 \in U_{\sigma}$, where $\*x_0$ and $\*y_0$ achieve the maximum in~\eqref{eq-def-dt}. 
Consider each transition step $(X, Y) \rightarrow (X', Y')$. 
Suppose $X = (u_x,i_x)$ and $Y = (u_y,i_y)$.
Then $X' = (u'_x,i'_x)$ and $Y' = (u'_y,i'_y)$ are generated as follows:
\begin{itemize}
\item sample $v \in V \setminus \Lambda$ uniformly at random, set $u'_x = u'_y = v$;	
\item sample $(i'_x,i'_y)$ from the optimal coupling of $\mu_{v}^{\sigma, u_x \gets i_x}$ and $\mu_{v}^{\sigma, u_y \gets i_y}$, where $v = u'_x = u'_y$.
\end{itemize}
It is easy to verify that $\+C$ is a coupling of Markov chain $Q_\sigma$. 
%Since $Q_\sigma$ is reversible with respect to $\pi$ and $Y_0 \sim \pi$, thus $Y_t \sim \pi$ for all $t \geq 0$.
%Note that $X_t \sim Q^t_{\sigma}(x_0,\cdot)$.
By \Cref{proposition-Markov-coupling}, we have
\begin{align}
\label{eq-bound-dtv-by-coupling}
\forall t \geq 1, \quad
d(t) = \max_{\*x_0,\*y_0 \in U_{\sigma}}\DTV{Q_\sigma^t(\*x_0,\cdot)}{Q_\sigma^t(\*y_0,\cdot)} \leq \Pr[\+C]{X_t \neq Y_t}.
\end{align}
Hence, we only need to bound the right-hand-side of~\eqref{eq-bound-dtv-by-coupling}.

Denote $X_t = (\Xvtx_t,\Xcol_t)$ and $Y_t = (\Yvtx_t,\Ycol_t)$. By the definition of the coupling procedure $\+C$, it holds that $\Xvtx_t=\Yvtx_t$ for all $t \geq 1$,  and
\begin{align}
\label{eq-same-vtx}
\forall t \geq 1, u \in V \setminus \Lambda,\quad  \Pr[\+C]{\Xvtx_t = \Yvtx_t = u} = \frac{1}{n-\abs{\Lambda}}.
\end{align}
For any $t \geq 1$, we define a column vector $e_t \in \mathbb{R}_{\geq 0}^{V \setminus \Lambda}$ such that
\begin{align*}
\forall u \in V \setminus \Lambda, \quad e_t(u) \triangleq \Pr[\+C]{\Xvtx_t = \Yvtx_t = u \land \Xcol_t \neq \Ycol_t}.	
\end{align*}
Then $d(t) \leq \Pr[\+C]{X_t \neq Y_t} = \sum_{u\in V \setminus \Lambda}e_t(u) = \norm{e_t}_1$ for all $t\geq 1$.
By~\eqref{eq-same-vtx}, we have 
\begin{align}
\label{eq-base}
\forall u \in V \setminus \Lambda, \quad e_1(u) \leq \Pr[\+C]{\Xvtx_t = \Yvtx_t = u} = \frac{1}{n-|\Lambda|}.
\end{align}

Recall $\Omega^\sigma_u \triangleq \{i \in [q] \mid \mu^\sigma_u(i) > 0\}$ for each $u \in V \setminus \Lambda$, 
and the state space of the random walk $Q_{\sigma}$ is $U_\sigma = \left\{(u,i) \mid u \in \overline{\Lambda} \land i \in \Omega^\sigma_u \right\}$. For any $t\geq 2$, we have for all $u \in V \setminus \Lambda$,
\begin{align*}
e_t(u) &= \Pr[\+C]{\Xvtx_t = \Yvtx_t = u \land \Xcol_t \neq \Ycol_t}\notag\\
&=\sum_{v \in V \setminus \Lambda}\sum_{\substack{i,j \in \Omega^\sigma_v\\ i \neq j}} \Big(\Pr[\+C]{\Xvtx_t = \Yvtx_t = u \land \Xcol_t \neq \Ycol_t \mid \Xvtx_{t-1}= \Yvtx_{t-1} = v \land \Xcol_{t-1} = i \land \Ycol_{t-1}=j }\\
&\qquad\qquad\qquad\quad \times \Pr[\+C]{X_{t-1}=(v,i) \land Y_{t-1}=(v,j)}\Big)\\
&=\sum_{v \in V \setminus \Lambda}\sum_{\substack{i,j \in \Omega^\sigma_v\\i \neq j}}\frac{1}{n-|\Lambda|}\DTV{\mu_u^{\sigma, v \gets i}}{\mu_u^{\sigma, v \gets j} }\Pr[\+C]{X_{t-1}=(v,i) \land Y_{t-1}=(v,j)}.
\end{align*}
The first equality is obtained from the chain rule, 
together with the facts that for $t\ge 2$, (1) $\Pr[\+C]{\Xvtx_{t-1}=\Yvtx_{t-1}} = 1$; 
(2) $\Xcol_t\neq \Ycol_t$ only if $\Xcol_{t-1} \neq \Ycol_{t-1}$; 
(3) $\Xcol_{t-1},\Ycol_{t-1} \in \Omega^\sigma_v$ if $\Xvtx_{t=1}=\Yvtx_{t-1}=v$ since $Q_\sigma$ is a random walk over $U_\sigma$. 
The last equality is obtained using the definition of the coupling $\+C$.
It holds because $\Xvtx_t=\Yvtx_t$ are sampled from $V \setminus \Lambda$ uniformly at random 
and $\Xcol_t,\Ycol_t$ are sampled from the optimal coupling between $\mu_u^{\sigma, v \gets i}$ and $\mu_u^{\sigma, v \gets j}$. 
By the definition of the influence matrix $\Psi^\sigma_\mu$ in~\eqref{eq-def-psi} and the definition of the matrix $A$ in~\eqref{eq-def-A}, 
we have that for any $u,v \in V \setminus \Lambda$, any $i,j \in \Omega^\sigma_u$ such that $i \neq j$, it holds that
\begin{align*}
\frac{1}{n-|\Lambda|}\DTV{\mu_u^{\sigma, v \gets i}}{\mu_u^{\sigma, v \gets j} } \leq A(u,v)	.
\end{align*}
Hence, for any $t\geq 2$, we have that for all $u \in V \setminus \Lambda$, $e_t(u)$ can be bounded by
\begin{align}
e_t(u) &\leq \sum_{v \in V \setminus \Lambda}\sum_{\substack{i,j \in \Omega^\sigma_v\\i \neq j}}A(u,v)\Pr[\+C]{X_{t-1}=(v,i) \land Y_{t-1}=(v,j)}%\notag\\
%&=
=
\sum_{v \in V \setminus \Lambda}A(u,v)e_{t-1}(v) = (Ae_{t-1})(u). \label{eq-induction}
\end{align}
All $e_t$ are non-negative vectors and $A$ is a non-negative matrix. 
Combining~\eqref{eq-bound-dtv-by-coupling}, ~\eqref{eq-base} and~\eqref{eq-induction}, we have that for any $t\ge 1$,
\begin{align*}
d(t) &\leq \norm{e_t}_1 \leq \norm{A^{t-1}e_1}_1	\leq \norm{A^{t-1}}_1 \norm{e_1}_1 \leq \norm{A^{t-1}}_1.\qedhere
\end{align*}
\end{proof}

Now, we are ready to prove \Cref{lemma-lambda2-Q}.
\begin{proof}[Proof of \Cref{lemma-lambda2-Q}]
Recall $Q_\sigma$ is a random walk over $U_\sigma$, $\pi$ is defined in~\eqref{eq-def-pi}.
Since $Q_\sigma$ is reversible with respect to $\pi$ and the support of $\pi$ is $U_\sigma$, by \Cref{proposition-markov-mat} and \Cref{lemma-coupling}, for any $t\geq 1$,
\begin{align*}
\abs{\lambda_2(Q_\sigma)}^t \leq d(t) \leq \norm{A^{t-1}}_1.	
\end{align*}
We may assume that $\lambda_2(Q_\sigma) > 0$, as otherwise \Cref{lemma-lambda2-Q} holds trivially. We have
\begin{align*}
\forall t \geq 1, \quad \lambda_2(Q_\sigma)^{\frac{t}{t-1}} \leq \norm{A^{t-1}}_1^{\frac{1}{t-1}}.	
\end{align*}
Let $t \rightarrow \infty$ in both sides, we have
\begin{align*}
\lambda_2(Q_\sigma) &= \lim_{t \to \infty }\lambda_2(Q_\sigma)^{\frac{t}{t-1}} \leq \lim_{t \to \infty}\norm{A^{t-1}}_1^{\frac{1}{t-1}} = \rho(A),
\end{align*}
where the last equality holds due to \Cref{proposition-radius}.
Note that  if $\lambda \in \mathbb{C}$ is an eigenvalue of $(\Psi_\mu^{\sigma})^T$, then $\lambda+1$ is an eigenvalue of $(\Psi_\mu^{\sigma})^T+I$, and $\abs{\lambda+1}\leq \abs{\lambda} + 1$.
By the definition of $A$, we have
\begin{align*}
\lambda_2(Q_\sigma) \leq \rho(A) &= \frac{1}{n-\abs{\Lambda}} \rho \tp{\tp{\Psi_\mu^{\sigma}}^T + I} \leq \frac{\rho\tp{\tp{\Psi_\mu^{\sigma}}^T}+1}{n-\abs{\Lambda}} = \frac{\rho(\Psi_\mu^\sigma)+1}{n-\abs{\Lambda}}.\qedhere
\end{align*}
\end{proof}
  
  %\section{Another proof of \Cref{lemma-lambda2-Q} (need revise)	}

\section{Rapid mixing for list colourings}

An instance of the list colouring is a pair $(G,\*L)$ where $G=(V,E)$ is a simple undirected
graph and $\*L = \set{L(v)\mid v\in V}$ is a collection of \emph{colour lists} associated to each
vertex $v\in V$. 
A proper list colouring $X$ assigns each vertex $v \in V$ a colour $X_v \in L(v)$ such that $X_u \neq X_v$ for all $\{u,v\} \in E$.
Let $\Omega_{G,\*L}$ denote the set of proper list colourings and $\mu_{G,\*L}$ denote the uniform distribution over $\Omega_{G,\*L}$.
 
The Glauber dynamics on $(G,\*L)$ is defined as follows. The chain starts from an arbitrary proper list colouring $X \in \Omega_{G,\*L}$. In each step, the chain does the following:
\begin{itemize}
\item pick a vertex $v \in V$ uniformly at random;
\item update $X_v$ by a uniformly at random colour from $L(v) \setminus \{X_u \mid \{v,u\} \in E \}$.	
\end{itemize}

We prove the following rapid mixing result for list colourings.
\begin{theorem}\label{thm-2D+1.76D}
	Let $(G=(V,E),\*L)$ be an instance of list colouring where $\*L=\set{L(v)\mid v\in V}$. Let $\Delta \geq 3$ be the maximum degree of $G$ and $\delta>0$ be a constant. If $G$ is triangle-free and for every $v\in V$, it holds that 
		\begin{align}
		\label{eq-proof-1.7D}
			\abs{L(v)}-\deg_G(v)\ge (\alpha^*+\delta-1)\Delta,
		\end{align}
		then the Glauber dynamics on $(G,\*L)$ satisfies
%		\[
%		\tmix(\eps)
%			\le 9^{1+t}\e^{\tp{2+\frac{2\delta}{\alpha^*}}(1+t)}n^{1+2t}\log\tp{\frac{M}{\eps}}=O\tp{n^{1+2t}\log\tp{\frac{M}{\eps}}},
%		\]
		\begin{align*}
		\tmix(\eps) \leq \tp{9\e^5 n}^{1 + 9/\delta}\cdot\log\tp{\frac{M}{\eps}}.	
		\end{align*}
where $M \defeq \prod_{v \in V}\abs{L(v)}$.
\end{theorem}
Note that \Cref{theorem-colouring} is a corollary of \Cref{thm-2D+1.76D}, in which $M = q^n$.

In order to prove~\Cref{thm-2D+1.76D}, we define a partial order $\preceq$ among list-colouring instances. Let $(G'=(V',E'),\*L')$ and $(G=(V,E),\*L)$ be two list colouring instances where $\*L'=\set{L'(v)\mid v\in V'}$ and $\*L=\set{L(v)\mid v\in V}$. We say $(G',\*L')\preceq (G,\*L)$ if there exists a vertex $v\in V$ satisfying
\begin{itemize}
	\item $G'=G[V\setminus\set{v}]$;
	\item for every $u\in \Gamma_G(v)$, it holds that $L'(u)\subseteq L(u)$ and $\abs{L(u)\setminus L'(u)}\le 1$;
	\item for every $u\in V'\setminus \Gamma_G(v)$, it holds that $L'(u)=L(u)$.
\end{itemize}

Here, $\Gamma_G(v)$ denotes the neighbourhood of $v$ in graph $G$.
We remark that in the definition above, for each $u\in \Gamma_G(v)$, 
we can rewrite the requirement as $L'(u)=L(u)\setminus\set{c}$ for some colour $c$. 
This colour $c$ is not necessarily in $L(u)$ (in which case $L'(u)=L(u)$ and can be distinct for different $u\in \Gamma_G(v))$. 

Intuitively, $(G',\*L')\preceq (G,\*L)$ means that one can obtain $(G',\*L')$ from $(G,\*L)$ by removing one vertex $v$ and change the colour lists of the neighbours of $v$ by removing at most one color. We call a family of list-colouring instances $\+L$ \emph{downward closed} if for every $(G,\*L)\in \mathscr{L}$ and every $(G',\*L')$ such that $(G',\*L')\preceq (G,\*L)$, we have $(G',\*L')\in \mathscr{L}$.

The \emph{downward closure} of an instance $(G,\*L)$ is the minimum downward closed family of instances containing $(G,\*L)$.

Consider the following condition for a family of list colouring instances $\^L$.

\begin{condition}\label{cond:main}
Let $\chi > 0$, $0<\eps_1<1$ and $\eps_2>0$. It holds that
\begin{itemize}
	\item the maximum degree of instances in $\^L$ is at most $\chi$;
	\item for any $(G=(V, E),\*L) \in \^L$, a proper list colouring exists, and for any vertex $v \in V$ satisfying $\deg_G(v) \leq \chi - 1$, it holds that
\begin{align}
\label{eq-marginal-up-2}
\forall c \in L(v):\quad
\mu_{v,(G,\*L)}(c) \leq \frac{\eps_1}{\deg_G(v)};
\end{align}
for any vertex $v \in V$, it holds that
\begin{align}
\label{eq-marginal-up-1}
\forall c \in L(v):\quad
\mu_{v,(G,\*L)}(c) \leq \frac{1}{\eps_2\chi+1}.
\end{align}	
\end{itemize}
\end{condition}

We have the following theorem.

\begin{theorem}\label{thm-family-main}
Let $0<\eps_1< 1$ and $\eps_2>0$ be two constants.
The following result holds for any $\chi > 0$.
Let $\^L$ be a downward closed family of list-colouring instances satisfying \Cref{cond:main} with parameters $\chi$, $\eps_1$ and $\eps_2$. For any $(G=(V,E),\*L)\in\^L$, the mixing time of Glauber dynamics satisfies
	\[
	\tmix(\eps)\le  \tp{9\e^{\frac{2}{\eps_2}}}^{\tp{1+\frac{1}{(1-\eps_1)\eps_2}}} n^{1+\frac{2}{(1-\eps_1)\eps_2}}	  \cdot\log\tp{\frac{M}{\eps}},
	\]
	where $M =\prod_{v \in V}\abs{L(v)}$.
\end{theorem}
\Cref{thm-2D+1.76D} is actually a corollary of \Cref{thm-family-main} via verifying \Cref{cond:main}.
We will prove~\Cref{thm-family-main} first. The proof of \Cref{thm-2D+1.76D} is deferred to  \Cref{section-proof-colouring}.

\subsection{Analysis of mixing time}
In the following, we assume $\mathscr{L}$ is downward closed and satisfies \Cref{cond:main}. Let $\chi>0$, $0<\eps_1<1$ and $\eps_2>0$ be the parameters promised by \Cref{cond:main}. 

For any list colouring instance $(G,\*L)$ where $G=(V,E)$, recall $\mu_{G,\*L}$ is the uniform distribution over all proper list colourings. 
Define the matrix $R_{G,\*L} \in \mathbb{R}_{\geq 0}^{V \times V}$ by
\begin{align}
\label{eq-def-recur}
\forall u,v \in V,\quad
R_{G,\*L}(u, v) = \max_{c_1,c_2 \in L(u)} \DTV{\mu_{v,(G,\*L)}^{u \gets c_1}}	{\mu_{v,(G,\*L)}^{u \gets c_2}},
\end{align}
where for $c=c_1$ or $c_2$, $\mu_{v,(G,\*L)}^{u \gets c}$ denotes the marginal distribution on $v$ projected from $\mu_{G,\*L}$ conditional on the colour of $u$ is fixed as $c$.
%and $\mu_{v,(G,\*L)}^{u \gets c_2}$ is defined in a similar way. 
The matrix $R$ is essentially the same as the influence matrix $\Psi^{\sigma_\Lambda}_\mu$ in~\eqref{eq-def-psi}, 
except that in the case of $u=v$, $R_{G,\*L}(v, v) = 0$ if and only if $\abs{L(v)} = 1$ (thus $c_1 = c_2$).
Namely, 
\begin{align*}
R_{G,\*L}(v, v)	= \max_{c_1,c_2 \in L(v)} \DTV{\mu_{v,(G,\*L)}^{v \gets c_1}}	{\mu_{v,(G,\*L)}^{v \gets c_2}} = \one{\abs{L(v)} > 1}.
\end{align*}

Roughly speaking, each entry $R_{G,\*L}(u,v)$ is the influence of $u$ on $v$ given two different colours of $u$.
The key to apply \Cref{theorem-main} is to bound the total influence of $u$ on all other vertices.

\begin{lemma}
\label{lemma-decay}
For any instance $(G=(V,E),\*L) \in \mathscr{L}$,
\begin{align*}
      \forall u\in V,\quad \sum_{v\in V:v\neq u} R_{G,\*L}(u,v) \leq \min\left\{\tp{1-\frac{1}{3\e^{1/\eps_2}}}(\abs{V}-1), \frac{1}{(1-\eps_1)\eps_2}\right\}.
\end{align*}
%where the big-O notation hides a universal constant factor. 
\end{lemma}

We first use \Cref{lemma-decay} to prove the main theorem for list colouring (\Cref{thm-family-main}).
Then we prove \Cref{lemma-decay} in \Cref{section-easy-coupling} and \Cref{section-recursive-coupling}. 

To prove \Cref{thm-family-main}, we will also need the following notion of pinning.
\begin{definition}[instance induced by pinning]
\label{definition-pin}
Let $(G=(V,E),\*L)$ be a list colouring instance. Let $\Lambda \subseteq V$ be a subset of vertices and $\sigma \in \otimes_{v \in \Lambda}L(v)$	a partial colouring on $\Lambda$. Define $\pin_{G,\*L}(\Lambda,\sigma) = (\widetilde{G},\widetilde{\*L})$ as the induced list colouring instance after the pinning $\sigma$, where $\widetilde{G}=G[V \setminus \Lambda ]$ is the subgraph of $G$ induced by $V \setminus \Lambda$, and $\widetilde{\*L} = \{\widetilde{L}(v) \mid v \in V \setminus \Lambda\}$ is defined by for all $v \in V \setminus \Lambda$,
\begin{align*}
\widetilde{L}(v) = L(v) \setminus \left\{ \sigma_u \mid u \in \Lambda \land \{u,v\} \in E \right\}.
\end{align*}
\end{definition}

It is clear that for any $\Lambda$ and $\sigma$, $\pin_{G,\*L}(\Lambda,\sigma)$ is in the downward closure of  $(G,\*L)$.

Now, we are ready to prove \Cref{thm-family-main}.
\begin{proof}[Proof of \Cref{thm-family-main}]
%Suppose $\delta \geq 0.3$. For all $(G = (V,E),\*L) \in \mathscr{L}(\delta)$, it holds that $\abs{L(v)} > 2.05\Delta + 3$ for all $v \in V$. In this case, it is well-known that the mixing time of Glauber dynamics is at most $ \frac{1.05}{0.05}n\log \frac{n}{\epsilon} = 21n\log \frac{n}{\epsilon}$~\cite{jerrum1995very,bubley1997path}. Note that $M \geq 2^n$. We have
%\begin{align*}
%\tmix(\epsilon) \leq  21n\log \frac{n}{\epsilon} \leq 100n^2 \log\frac{1}{\epsilon} \leq 10^{2+10/\delta} n^{1+10/\delta}	\left(  \log \frac{M}{\epsilon}\right).
%\end{align*}
%

%We turn to the main case $0<\delta < 0.3$. 
It suffices to verify that every $(G,\*L) \in \^L$ is $(C,\eta)$-spectrally independent, which implies the theorem by \Cref{theorem-main}. 
Fix a list colouring instance $(G= (V,E),\*L) \in \^L$. Fix a subset $\Lambda \subseteq V$ with $\abs{\Lambda} \leq n - 2$ and a feasible partial colouring $\sigma_\Lambda \in \otimes_{v \in \Lambda}L(v)$. Let $(\widetilde{G},\widetilde{\*L}) = \pin_{G,\*L}(\Lambda,\sigma_\Lambda)$, where $\widetilde{G}=G[V \setminus \Lambda]$ and $\widetilde{\*L} = \{\widetilde{L}(v) \mid v \in V \setminus \Lambda\}$.
%, and $\pin_\cdot(\cdot)$ is defined in \Cref{definition-pin}.
Note that for any $u \in V \setminus \Lambda$, $\widetilde{L}(u)$ contains precisely the feasible colours for $u$ conditional on $\sigma_\Lambda$.
Then, by the definition of $\Psi^{\sigma_\Lambda}_\mu$ in \eqref{eq-def-psi}, 
\begin{align*}
\forall u,v \in V \setminus \Lambda \text{ with } u\neq v,\quad
\Psi^{\sigma_\Lambda}_\mu(u,v) 	&= \max_{c_1,c_2 \in \widetilde{L}(u) }\DTV{\mu_{v,(G,\*L)}^{\sigma_\Lambda, u \gets c_1} }{ \mu_{v,(G,\*L)}^{\sigma_\Lambda, u \gets c_2}}\\
\tag*{(by \Cref{definition-pin})}\quad&= \max_{c_1,c_2 \in \widetilde{L}(u) }\DTV{\mu_{v,(\widetilde{G},\widetilde{\*L})}^{ u \gets c_1} }{ \mu_{v,(\widetilde{G},\widetilde{\*L})}^{u \gets c_2}}\\
&= R_{\widetilde{G},\widetilde{\*L}}(u, v).
\end{align*}
Also by the definition of $\Psi^{\sigma_\Lambda}_\mu$, for any $v \in V \setminus \Lambda$, it holds that $\Psi^{\sigma_\Lambda}_\mu(v,v) = 0$. Since $\^L$ is downward closed, $(\widetilde{G},\widetilde{\*L}) \in \^L$. By \Cref{lemma-decay},
\begin{align*}
\norm{\Psi_\mu^{\sigma_\Lambda}}_{\infty} &= \max_{u \in V \setminus \Lambda}\sum_{v\in V \setminus \Lambda} \Psi_{\mu}^{\sigma_\Lambda}(u,v)  = \max_{u \in V \setminus \Lambda}\sum_{v\in V \setminus \Lambda:v \neq u} R_{\widetilde{G},\widetilde{\*L}}(u,v)\\
&\leq \min\left\{\tp{1-\frac{1}{3\e^{1/\eps_2}}}(n - \abs{\Lambda}-1), \frac{1}{(1-\eps_1)\eps_2}\right\}.
\end{align*}
Hence, the list colouring instance $(G,\*L) \in \^L$ satisfies bound one-to-all influence condition in \Cref{corollary-main} with $C= \frac{1}{(1-\eps_1)\eps_2}$ and $\eta = 1-\frac{1}{3\e^{1/\eps_2}}$. By \Cref{corollary-main}, Glauber dynamics on $(G,\*L)$ has mixing time
\begin{align*}
\tmix(\epsilon) \leq  \frac{n^{1+\frac{2}{(1-\eps_1)\eps_2}}}{\tp{\frac{1}{3\e^{1/\eps_2}}}^{2+\frac{2}{(1-\eps_1)\eps_2}}}\cdot \log\tp{\frac{1}{\epsilon \mu_{\min}}}
 \leq
 	 \tp{9\e^{\frac{2}{\eps_2}}}^{\tp{1+\frac{1}{(1-\eps_1)\eps_2}}} n^{1+\frac{2}{(1-\eps_1)\eps_2}}	  \cdot\log\tp{\frac{M}{\eps}},
\end{align*}
where the last inequality holds because $\frac{1}{\mu_{\min}} \leq M =\prod_{v \in V}\abs{L(v)}$.
\end{proof}	

The two upper bounds in \Cref{lemma-decay} are proved in \Cref{section-easy-coupling} and \Cref{section-recursive-coupling} respectively.

\subsection{An easy coupling analysis}
\label{section-easy-coupling}
We now prove the first part of \Cref{lemma-decay}, namely,
\begin{lemma}
\label{lemma-decay-1}
Let $\mathscr{L}$ be a downward closed family of list colouring instances satisfying \Cref{cond:main} with parameters $\chi>0$, $0<\eps_1<1$ and $\eps_2>0$. For any instance $(G=(V,E),\*L)\in\mathscr{L}$, it holds that
\begin{align*}
      \forall u\in V,\quad \sum_{v\in V:v\neq u} R_{G,\*L}(u,v) \leq \tp{1-\frac{1}{3\e^{1/\eps_2}}}(\abs{V}-1).
\end{align*}
%where the big-O notation hides a universal constant factor. 
\end{lemma}
%we prove that for any $(G=(V,E),\*L) \in \mathscr{L}(\delta)$,
%\begin{align}
%\label{eq-proof-1}
%      \forall u\in V,\quad \sum_{v\in V} R_{G,\*L}(u,v) \leq \frac{14}{15}(\abs{V}-1).
%\end{align}
To prove \Cref{lemma-decay-1}, we need the following well-known recursion of list colouring.
\begin{proposition}[\cite{GK12,lu2013improved,GKM15}]
\label{proposition-recursion}
Let $(G=(V,E),\*L) \in \^L$ be a list colouring instance. Let $v_1,v_2,\ldots,v_m$ denote the neighbours of $v$ in $G$. Let $c \in L(v)$ be a colour. Let $G_v$ be the subgraph of $G$ induced by $V \setminus \{v\}$. 
For each $ 1\leq i \leq m$, define a colour list $\*L_{i,c}=\{L_{i,c}(u) \mid u \in V \setminus \{v\} \}$, 
where $L_{i,c}(u) = L(u) \setminus \{c\}$ for all $u=v_j$ and $j < i$, and $L_{i,c}(u) = L(u)$ for other vertices. It holds that for any $c \in L(v)$,
\begin{align*}
\mu_{v,(G,\*L)}(c)	&= \frac{\prod_{i=1}^m\tp{1-\mu_{v_i,(G_v,\*L_{i,c})}(c)}}{\sum_{c' \in L(v)} \prod_{i=1}^m\tp{1-\mu_{v_i,(G_v,\*L_{i,c'})}(c')} }.	
\end{align*}
\end{proposition}

We first derive upper and lower bounds for marginal probabilities from \Cref{cond:main}.

\begin{lemma}
\label{lemma-up-low}
Let $\mathscr{L}$ be a downward closed family of list colouring instances satisfying \Cref{cond:main} with parameters $\chi>0$ and $0<\eps_1<1$ and $\eps_2>0$. For any instance $(G=(V,E),\*L)\in\mathscr{L}$, it holds that
\begin{align*}
\forall c \in L(v), \quad \frac{1}{\e^{1/\eps_2}\abs{L(v)}} \leq \mu_{v,(G,\*L)}(c) \leq \frac{1}{\eps_2\chi+1} .
\end{align*}
\end{lemma}

\begin{proof}
The upper bound is directly from \Cref{cond:main}. So we only need to prove the lower bound.

Fix an instance $(G,\*L) \in \^L$.
Since $\mathscr{L}$ is downward closed, each instance $(G_v,\*L_{i,c})\in\mathscr{L}$, where  $(G_v,\*L_{i,c})$ is defined in \Cref{proposition-recursion}. By the recursion in \Cref{proposition-recursion}, we have
%\ctodo{Maybe we should put this recursion in the Prelim.}
\begin{align*}
\mu_{v,(G,\*L)}(c)	&= \frac{\prod_{i=1}^m\tp{1-\mu_{v_i,(G_v,\*L_{i,c})}(c)}}{\sum_{c' \in L(v)} \prod_{i=1}^m\tp{1-\mu_{v_i,(G_v,\*L_{i,c'})}(c')} }
\geq\frac{\tp{1-\frac{1}{\eps_2\chi+1}}^\chi}{\abs{L(v)}} \geq \frac{1}{\e^{1/\eps_2}\abs{L(v)}}.
\end{align*}
This proves the lower bound.
\end{proof}
Now, we are ready to prove~\Cref{lemma-decay-1}. 

\begin{proof}[Proof of \Cref{lemma-decay-1}]
Consider the list colouring instance $(G=(V,E),\*L)$. 
%For any $v \in V$, we use $\Gamma_G(v)$ to denote the neighbours of $v$ and let $\Gamma^+_G(v) = \Gamma_G(v) \cup \{v\}$.
Fix a vertex $u$ and two colours $c_1,c_2 \in L(u)$. 
Define a list colouring instance $\+L_1=(G_u, \*L_1)= \pin_{G,\*L}(\{u\},c_1)$, where $G_u$ is the subgraph of $G$ induced by $V \setminus \{u\}$ and $\*L_1 = \{L_1(w) \mid w \in V \setminus \{u\}\}$.
%, and $\pin_{\cdot}(\cdot)$ is defined in \Cref{definition-pin}.
%Define the list colouring instance $\+L_1=(G_u, \*L_1)$, where $G_u$ be the subgraph of $G$ induced from $V \setminus \{u\}$, and  $\*L_1 = \{L_1(w) \mid w \in V \setminus \{u\}\}$ satisfies $L_1(w) = L(w) \setminus \{c_1\}$ for all $w \in \Gamma(u)$  and $L_1(w) = L(w)$ for all $w \in V \setminus \Gamma^+(u)$.
Define a list colouring instance $\+L_2=(G_u, \*L_2)= \pin_{G,\*L}(\{u\},c_2)$, where $\*L_2 = \{L_2(w) \mid w \in V \setminus \{u\}\}$.
%Define the list colouring instance $\+L_2=(G_u, \*L_2)$, where $\*L_2 = \{L_2(w) \mid w \in V \setminus \{u\}\}$ satisfies $L_2(w) = L(w) \setminus \{c_2\}$ for all $w \in \Gamma(u)$  and $L_2(w) = L(w)$ for all $w \in V \setminus \Gamma^+(u)$. 
Then
\begin{align*}
\forall v\neq u, \quad \mu_{v,(G,\*L)}^{u \gets c_1}(\cdot) = \mu_{v,\+L_1}(\cdot), \quad 	\mu_{v,(G,\*L)}^{u \gets c_2}(\cdot) = \mu_{v,\+L_2}(\cdot).
\end{align*}
Since $(G,\*L) \in \mathscr{L}$ and $\mathscr{L}$ is downward closed, it holds that both $\+L_1,\+L_2 \in \mathscr{L}$.
By \Cref{lemma-up-low}, for any $v \neq u$,
\begin{align*}
\forall c \in L_1(v): \quad 	 &\mu_{v,\+L_1}(c) \geq \frac{1}{\e^{1/\eps_2}\abs{L_1(v)}}\\
\forall c \in L_2(v): \quad 	 &\mu_{v,\+L_2}(c) \geq \frac{1}{\e^{1/\eps_2}{\abs{L_2(v)}}}.
\end{align*}
On the other hand, since  $\+L_1,\+L_2 \in \mathscr{L}$, for any $v \in V$, it holds that $\abs{L_1(v)} \geq 2$ and $\abs{L_2(v)} \geq 2$ (otherwise, the upper bound for the marginals in \Cref{cond:main} cannot hold).
By the definitions of $\+L_1$ and $\+L_2$, it holds that $\abs{L_1(v) \cap L_2(v)} \geq \min\{\abs{L_1(v), L_2(v)}\} - 1$ and $\big|\abs{L_1(v)} - \abs{L_2(v)}\big|\leq 1$. 
Thus, we can couple $\mu_{v,\+L_1}(\cdot)$ and  $\mu_{v,\+L_2}(\cdot)$ with success probability at least 
\begin{align*}
&\sum_{c \in L_1(v) \cap L_2(v)}	\min\left\{\frac{1}{\e^{1/\eps_2}\abs{L_1(v)}},\frac{1}{\e^{1/\eps_2}\abs{L_2(v)}}\right\}\geq \frac{\min\{\abs{L_1(v), L_2(v)}\} - 1}{\e^{1/\eps_2} \max\{\abs{L_1(v)},\abs{L_2(v)}\}}\\
\geq \,& \frac{1}{e^{1/\eps_2}} \cdot \frac{\min\{\abs{L_1(v)}, \abs{L_2(v)}\} - 1}{ \min\{\abs{L_1(v)},\abs{L_2(v)}\}+1} \geq \frac{1}{3\e^{1/\eps_2}}.
\end{align*}
By the coupling inequality (\Cref{proposition-coupling}), we have for any $c_1,c_2 \in L(u)$ and any $v \neq u$,
\begin{align*}
\DTV{\mu_{v,(G,\*L)}^{u \gets c_1}}	{\mu_{v,(G,\*L)}^{u \gets c_2}}	= \DTV{\mu_{v,\+L_1}}	{\mu_{v,\+L_2}} \leq 1 - \frac{1}{3\e^{1/\eps_2}}.
\end{align*}
By the definition of $R_{G,\*L}$, we have that
\begin{align*}
\sum_{v \in V:v \neq u} R_{G,\*L}(u,v) & \leq \tp{1-\frac{1}{3\e^{1/\eps_2}}}\tp{\abs{V}-1}. \qedhere
\end{align*}
\end{proof}
\subsection{Recursive coupling}
\label{section-recursive-coupling}
We then prove the second part of \Cref{lemma-decay}.
\begin{lemma}
\label{lemma-decay-2}
Let $\mathscr{L}$ be a downward closed family of list colouring instances satisfying \Cref{cond:main} with parameters $\chi>0$, $0<\eps_1<1$ and $\eps_2>0$. For any instance $(G=(V,E),\*L) \in \mathscr{L}$, it holds that
\begin{align*}
      \forall u\in V,\quad \sum_{v\in V: v\neq u} R_{G,\*L}(u,v) \leq \frac{1}{(1-\eps_1)\eps_2}.
\end{align*}
%where the big-O notation hides a universal constant factor. 
\end{lemma}
We use the following lemma to prove~\Cref{lemma-decay-2}.
\begin{definition}[self-avoiding walk (SAW)]
  A path $P = (v_1,v_2,\ldots,v_\ell)$ in a graph $G$ is called a \emph{self-avoiding walk} (SAW) if each $v_i$ and $v_{i+1}$ are adjacent and $v_i \neq v_j$ for all $i \neq j$. 
\end{definition}

\begin{lemma}
  \label{lemma-path}
  Let $\mathscr{L}$ be a downward closed family of list colouring instances satisfying \Cref{cond:main} with parameters $\chi>0$, $0<\eps_1<1$ and $\eps_2>0$. 
  For any instance $(G=(V,E),\*L) \in \mathscr{L}$ and any two vertices $u,v\in V$ with $u\ne v$, it holds that
\begin{align}
\label{eq-lemma-path}
 R_{G,\*L}(u,v) \leq \frac{1}{\eps_1\eps_2}\sum_{\substack{\text{SAW $P=(v_1,v_2,\ldots,v_{\ell})$}\\\text{$u=v_1$ and $v = v_{\ell}$}}} \prod_{k=1}^{\ell-1} \frac{\eps_1}{|\Gamma_G(v_k) \setminus \{v_i \mid i < k\} |},
\end{align}
where $\Gamma_G(v_k)$ is the neighbourhood of $v_k$ in $G$.
\end{lemma}
We remark that the denominator of each ratio in the RHS of~\eqref{eq-lemma-path} is positive because $v_{k+1} \in \Gamma_G(v_k) \setminus \{v_i \mid i < k\}$ for all $1\leq k \leq \ell-1$. %We now use \Cref{lemma-path} to prove~\Cref{theorem-main}.
\Cref{lemma-path} is proved in \Cref{section-proof-path} via a recursive coupling argument.

Now, we are ready to prove~\Cref{lemma-decay-2}.

\begin{proof}[Proof of \Cref{lemma-decay-2}]
%
%Fix a vertex $v \in V$. Note that $\deg_G(v) \leq \Delta$. 
%Also note that $R_{G,\*L}(v,v)\leq 1$ for all $v \in V$ due to~\eqref{eq-def-recur}.
%By \Cref{lemma-path}, we have
%\begin{align}
%\label{eq-proof-path-bound}
%\forall u,v \in V,\quad  R_{G,\*L}(u,v) &\leq \frac{1}{(\alpha^* - 1)(1-\delta/3)}\sum_{\substack{\text{SAW $P=(v_1,v_2,\ldots,v_{\ell})$}\\\text{$u=v_1$ and $v = v_{\ell}$}}} \prod_{k=1}^{\ell-1} \frac{1-\delta/3}{|\Gamma_G(v_k) \setminus \{v_i \mid i < k\} |}\notag\\
%(\text{by } \delta < 0.3)\quad&\leq .
%\end{align}
%Remark that the RHS of~\eqref{eq-proof-path-bound} is $\frac{1}{0.9(\alpha^*-1)} \geq 1 \geq R_{G,\*L}(v,v)$ if $u = v$.

%For any vertices $u,v \in V$, we use $P_{\ell}^{u \rightarrow v}$ to denote all the SAWs from $u$ to $v$ with length $\ell$.
Fix $(G=(V,E), \*L) \in \^L$.
For any vertex $u \in V$ and any integer $\ell \geq 1$, we use $P_{\ell}^u$ to denote the set of all SAWs from $u$ that contains $\ell$ vertices. Formally,
$P_{\ell}^{u} \triangleq \{P=(v_1,v_2,\ldots,v_\ell) \mid P \text{ is a SAW}, v_1= u \}.$
We claim that
\begin{align}
\label{eq-claim-path}
\forall u \in V,\ell \geq 1, \quad \sum_{\substack{\text{SAW }P =(v_1,v_2,\ldots,v_\ell) \in P_\ell^u }}\prod_{k=1}^{\ell - 1}\frac{\eps_1}{|\Gamma_G(v_k) \setminus \{v_i \mid i < k\} |} \leq \eps_1^{\ell-1}.
\end{align}
We now use~\eqref{eq-claim-path} to prove \Cref{lemma-decay-2}.
By~\Cref{lemma-path}, for any $u \in V$,
\begin{align*}
  \sum_{v \in V:v\ne u} R_{G,\*L}(u,v) 
  &\leq  \frac{1}{\eps_1\eps_2}\cdot\sum_{v \in V:v\ne u}\sum_{\substack{\text{SAW }P=(v_1,v_2,\ldots,v_{\ell})\\\text{$u=v_1$ and $v = v_{\ell}$}}} \prod_{k=1}^{\ell-1} \frac{\eps_1}{|\Gamma_G(v_k) \setminus \{v_i \mid i < k\} |}\\
  \tag{$\star$}
  &\le \frac{1}{\eps_1\eps_2}\cdot \sum_{\ell = 2}^{\infty}\sum_{\substack{\text{SAW }P =(v_1,v_2,\ldots,v_\ell) \in P_\ell^u }}\prod_{k=1}^{\ell - 1}\frac{\eps_1}{|\Gamma_G(v_k) \setminus \{v_i \mid i < k\} |}\\
  \tag{by~\eqref{eq-claim-path}}
  &\le \frac{1}{\eps_1\eps_2}\cdot \sum_{\ell=2}^{\infty}\eps_1^{\ell-1} = \frac{1}{(1-\eps_1)\eps_2},
\end{align*}
where $(\star)$ is due to the fact that $v\ne u$ implies all SAWs in consideration are of length at least $2$. This proves~\Cref{lemma-decay-2}.

We then prove~\eqref{eq-claim-path} by an induction on $\ell$.  
If $\ell = 1$, the LHS of~\eqref{eq-claim-path} is 1, thus~\eqref{eq-claim-path} holds trivially. Suppose~\eqref{eq-claim-path} holds for all $\ell \leq t$, we prove it for $\ell = t+1$. 
Let $P_{t}^{u \rightarrow v}$ denote the set of all SAWs from $u$ to $v$ that contains $t$ vertices. Formally,  
\begin{align*}
P_{t}^{u \rightarrow v} \triangleq \{P=v_1,v_2,\ldots,v_t \mid P \text{ is a SAW}, v_1= u, v_t =v \}.
\end{align*}   
Hence, $P_{t}^u = \bigcup_{v \in V}P_{t}^{u \rightarrow v}$. 
If $P \in P_{t+1}^u$ is a SAW such that $P=v_1,v_2,\ldots,v_t,v_{t+1}$, 
then the prefix $P' = v_1,v_2,\ldots,v_t$ is in the set $P_{t}^{u \rightarrow v_{t}}$ and $v_{t+1} \in \Gamma_G(v_{t}) \setminus \{v_i \mid i < t\}$,
and vice versa.
This implies that
\begin{align}
  \label{eq-saw}
  P_{t+1}^u = \bigcup_{v \in V}\left\{(P, w) \mid P =(v_1=u,v_2,\ldots,v_t=v) \in P_{t}^{u \rightarrow v}, w \in \Gamma_G(v) \setminus \{v_i \mid i < t\} \right\},	
\end{align}
where $(P,w)$ is the path obtained by appending $w$ at the end of the path $P$.  We have 
%We prove~\eqref{eq-claim-path} 
\begin{align*}
\forall u\in V,\qquad
&\sum_{\substack{\text{SAW }P =(v_1,v_2,\ldots,v_{t+1}) \in P_{t+1}^u }}\prod_{k=1}^{t}\frac{\eps_1}{|\Gamma_G(v_k) \setminus \{v_i \mid i < k\} |}\\
=&\, \sum_{v \in V}\sum_{\substack{\text{SAW }P = (v_1,v_2,\ldots,v_t) \in P_{t}^{u\rightarrow v}}}\sum_{w \in \Gamma_G(v) \setminus \{v_i \mid i < t\}} \prod_{k=1}^{t}\frac{\eps_1 }{|\Gamma_G(v_k) \setminus \{v_i \mid i < k\} |}\tag{by~\eqref{eq-saw}}\\
\leq&\,\eps_1\cdot \sum_{v\in V}\sum_{\substack{\text{SAW }P = (v_1,v_2,\ldots,v_t) \in P_{t}^{u\rightarrow v}}}\prod_{k=1}^{t-1}\frac{\eps_1}{|\Gamma_G(v_k) \setminus \{v_i \mid i < k\} |}\tag{$\star$}\\
=&\, \eps_1\cdot \sum_{\substack{\text{SAW }P =(v_1,v_2,\ldots,v_t)  \in P_t^u }}\prod_{k=1}^{t-1}\frac{\eps_1}{|\Gamma_G(v_k) \setminus \{v_i \mid i < k\} |}\\
\leq&\, \eps_1^{t}.\tag{by I.H.}
\end{align*}
The inequality $(\star)$ holds because $\Gamma_G(v) \setminus \{v_i \mid i < t\} = \Gamma_G(v_t) \setminus \{v_i \mid i < t\}$ (due to  $v_t = v$). 
We remark the $(\star)$ is an inequality rather than an equality because $\Gamma_G(v) \setminus \{v_i \mid i < t\}$ can be empty.
This proves~\eqref{eq-claim-path}.	
\end{proof}

\subsubsection{Influence bounds via recursion}
\label{section-proof-path}
Now we prove Lemma~\ref{lemma-path}. The proof technique is based on the ``recursive coupling''
 introduced by Goldberg, Martin and Paterson~\cite{GMP05}. 
 
We introduce some definitions.
Let $(G,\*L)$ be a list colouring instance, where $G=(V, E)$. 
Fix a vertex $u \in V$ and two colours $c_1,c_2 \in L(u)$. Let $w_1,w_2,\ldots,w_m$ denote the neighbours of $u$ in graph $G$, where $m =\deg_G(u)$. For any $0\leq k \leq m$, we define a list colouring instance $(G_u, \*L_{u, k}^{c_1,c_2})$: The graph $G_u = G[V \setminus \{u\}]$ is obtained by removing vertex $u$ from $G$. The colour list $\*L_{u,k}^{c_1,c_2}$ is obtained by removing the colour $c_1$ from the lists $L(w_\ell)$ for $\ell < k$, and removing the colour $c_2$ for the lists $L(w_\ell)$ for $\ell > k$. Formally, 
\begin{align}
\label{eq-def-new-list}
\forall v \in V \setminus \{u\}:\quad L_{u,k}^{c_1,c_2}(v) = \begin{cases}
L(v) \setminus \{c_1\} &\text{if } v \in \{w_1,w_2,\ldots,w_{k-1}\} \\
L(v) \setminus \{c_2\} &\text{if } v \in \{w_{k+1},w_{k+2},\ldots,w_m\} \\
L(v) &\text{if } v \not\in \Gamma_G(u) \text{ or } v = w_k.	
 \end{cases}
\end{align}

%Next, we define the following family of list colouring instances.
%\begin{definition}
%\label{definition-list colouring-2}
%Let $\chi \geq 0$ be an integer and $0< \delta < 0.3$ be a constant. 
%Define $\mathscr{F}(\delta,\chi)$ as the family of all list colouring instances $(G=(V,E),\*L)$ that satisfy
%\begin{itemize}
%\item $G$ is a triangle-free graph;
%\item the maximum degree $\Delta$ of $G$ satisfies $\Delta \leq \chi$;
%\item for any $v \in V$, 
%\begin{align*}
%\abs{L(v)} - \deg_G(v) \geq (\alpha^* + \delta - 1)\chi + 3, 
%\end{align*}
%where $\alpha^*\approx 1.763\ldots$ is the positive root of equation $x^x=\mathrm{e}$.
%\end{itemize}
%\end{definition}
\begin{lemma}
\label{lemma-recursive}
Let $\mathscr{L}$ be a downward closed family of list colouring instances satisfying \Cref{cond:main} with parameters $\chi>0$, $0<\eps_1<1$ and $\eps_2>0$. For any $(G=(V, E),\*L) \in \mathscr{L}$, the following result holds.
Fix a pair of vertices $u, v \in V$.
Let $w_1,w_2,\ldots,w_{\deg_G(u)}$ denote the neighbours of $u$ in $G$.
Let $c_1,c_2 \in L(u)$ be the colours achieving the maximum in~\eqref{eq-def-recur} (breaking ties arbitrarily). 
Then, %there exists $\gamma  = 1 -\frac{1}{\alpha} \exp\tp{\frac{1}{\alpha}(1+1/\beta)} > 0$ such that
\begin{align*}
R_{G,\*L}(u, v) \leq \begin{cases}
1 &\text{if $u=v$};\\
0 &\text{if $u$ and $v$ are disconnected in $G$};\\
\sum_{k = 1}^{\deg_G(u)}\alpha_k\cdot R_{G_u, \*L_{u,k}^{c_1,c_2}}(w_k,v) &\text{otherwise}.
 \end{cases},
\end{align*}
where for all $1\leq k \leq \deg_G(u)$,
\begin{align*}
\alpha_k = \min \tp{ \frac{\eps_1}{\deg_{G_u}(w_k)}, \frac{1}{\eps_2\chi+1}}.
\end{align*}
\end{lemma}

We remark that if $\deg_{G_u}(w_k) = 0$, then by convention we have $\frac{\eps_1}{\deg_{G_u}(w_k)} = \infty$ and thus  $\alpha_k = \frac{1}{\eps_2\chi+1}$.

Now we use \Cref{lemma-recursive} to derive \Cref{lemma-path} and defer the proof of \Cref{lemma-recursive} to \Cref{section-proof-recursive}.
\begin{proof}[Proof of \Cref{lemma-path}]
 Suppose $(G=(V,E),\*L)\in \^L$. It is clear that the instance $(G_u, \*L_{u,k}^{c_1,c_2})$ obtained from $(G,\*L)$ is also in $\^L$. Hence, we can use \Cref{lemma-recursive} recursively. This implies that for any  $(G,\*L)\in \^L$,  any $u,v \in V$,
\begin{align}
\label{eq-recur-proof}
R_{G,\*L}(u,v) \leq \sum_{\substack{\text{SAW $P=(v_1,v_2,\ldots,v_{\ell})$}\\\text{$u=v_1$ and $v = v_{\ell}$}}} \prod_{k=2}^{\ell} \min\tp{\frac{\eps_1}{\abs{\Gamma_G(v_k) \setminus \{v_i \mid i < k\} }}, \frac{1}{\eps_2\chi+1}}.
\end{align}

If $u$ and $v$ are disconnected, then $R_{G,\*L}(u,v)= 0$, and in this case,  the RHS of~\eqref{eq-lemma-path} is $0$ because there is no SAW from $u$ to $v$, thus ~\eqref{eq-lemma-path} holds. So in the following we assume that $u$ and $v$ are connected.
%Let $\chi' = \Delta_G$ denote the maximum degree of graph $G$. 
%Consider the family of list colouring instances $\mathscr{F}(\delta, \chi')$.
%Since $(G,\*L) \in \mathscr{L}(\delta)$, by~\Cref{definition-list colouring} and \Cref{definition-list colouring-2}, we have $(G,\*L) \in \mathscr{F}(\delta, \chi')$.
%We use~\eqref{eq-recur-proof} on family $\mathscr{F}(\delta, \chi')$ and instance $(G=(V,E),\*L)$.
%Since $u$ and $v$ are connected, the SAW between $u$ and $v$ exists and $\deg_G(u) > 0$. 

Our goal is to prove \eqref{eq-lemma-path}.
Comparing \eqref{eq-recur-proof} with \eqref{eq-lemma-path},
the main difference is the range of $k$ in the product.
We will trade the last factor of $\frac{1}{\eps_2\chi+1}$ for $k=\ell$ by a factor of $\frac{\eps_1}{\chi}$ for $k=1$,
with a loss of $\frac{1}{\eps_1\eps_2}$.

More precisely, by~\eqref{eq-recur-proof}, we have
\begin{align*}
R_{G,\*L}(u,v) &\overset{(\star)}{\leq} \frac{1}{\eps_2\chi+1}\sum_{\substack{\text{SAW $P=(v_1,v_2,\ldots,v_{\ell})$}\\\text{$u=v_1$ and $v = v_{\ell}$}}} \tp{\prod_{k=2}^{\ell-1} \frac{\eps_1}{\abs{\Gamma_G(v_k) \setminus \{v_i \mid i < k\} }} }\\
%\text{(by $\chi'=\Delta_G$)}\quad&=	\frac{1}{(\alpha^*-1)\Delta_G+3}\sum_{\substack{\text{SAW $P=(v_1,v_2,\ldots,v_{\ell})$}\\\text{$u=v_1$ and $v = v_{\ell}$}}} \tp{\prod_{k=2}^{\ell-1} \frac{1-\delta/3}{\abs{\Gamma_G(v_k) \setminus \{v_i \mid i < k\} }} }\\
\tag{as $0<\deg_G(v_1) \leq \chi$}\quad&\leq \frac{\chi}{\eps_1(\eps_2\chi+1)}\sum_{\substack{\text{SAW $P=(v_1,v_2,\ldots,v_{\ell})$}\\\text{$u=v_1$ and $v = v_{\ell}$}}} \tp{\prod_{k=1}^{\ell-1} \frac{\eps_1}{\abs{\Gamma_G(v_k) \setminus \{v_i \mid i < k\} }} }\\
 &\leq \frac{1}{\eps_1\eps_2}\sum_{\substack{\text{SAW $P=(v_1,v_2,\ldots,v_{\ell})$}\\\text{$u=v_1$ and $v = v_{\ell}$}}} \tp{\prod_{k=1}^{\ell-1} \frac{\eps_1}{\abs{\Gamma_G(v_k) \setminus \{v_i \mid i < k\} }} },
\end{align*}
where inequality $(\star)$ holds due to~\eqref{eq-recur-proof} and $\ell \geq 2$ (since $u \neq v$).
Note that in the formula above, it holds that $\abs{\Gamma_G(v_k) \setminus \{v_i \mid i < k\}} > 0$ for all $1\leq k \leq \ell - 1$ because $v_{k+1} \in \Gamma_G(v_k) \setminus \{v_i \mid i < k\}$.
This proves~\Cref{lemma-path}.
\end{proof}

\subsubsection{Establish recursion via coupling}
\label{section-proof-recursive}
Next, we prove \Cref{lemma-recursive}. Fix an instance
$(G=(V, E),\*L) \in \^L$. Fix a vertex $u \in V$. %Let $v_1,v_2,\ldots,v_{\deg_G(v)}$ denote the neighbours of $v$ in $G$.
Let $c_1,c_2 \in L(u)$ be the colours achieving  the maximum in~\eqref{eq-def-recur} (breaking ties arbitrarily).  Our goal is to bound 
\begin{align*}
R_{G,\*L}(u, v) = \max_{c_1,c_2 \in L(u)} \DTV{\mu_{v,(G,\*L)}^{u \gets c_1}}	{\mu_{v,(G,\*L)}^{u \gets c_2}}.
\end{align*}
If $u = v$, then $R_{G,\*L}(u, v) \leq 1 $. If $u$ and $v$ are disconnected in $G$, then $R_{G,\*L}(u, v) = 0 $. In the rest of this section, we assume $u \neq v$ and $u,v$ are connected in graph $G$.

Let $w_1,w_2,\ldots,w_m$ denote the neighbours of $u$ in $G$, where $m = \deg_G(u)$. We construct a graph $G'$ from $G$ as follows. 
We remove the vertex $u$ from the graph $G$, add $m$ new vertices $u_1,u_2,\ldots,u_m$, and then add $m$ new edges $\{u_i,w_i\}$. Finally, we define a set of colour lists $\*L' = \{L'(v)\mid v \in V \setminus \{u\} \cup \{u_1,u_2,\ldots,u_m\}\}$ as 
\begin{align*}
L'(v) \triangleq \begin{cases}
L(u) &\text{if } v \in \{u_1,u_2,\ldots,u_m\}\\
L(v) &\text{if } v \in V \setminus \{u\}.
 \end{cases}
\end{align*}
This defines a new list colouring instance $(G',\*L')$. Figure~\ref{figure-ex} gives a small example.

 \begin{figure}[htb]
 \includegraphics{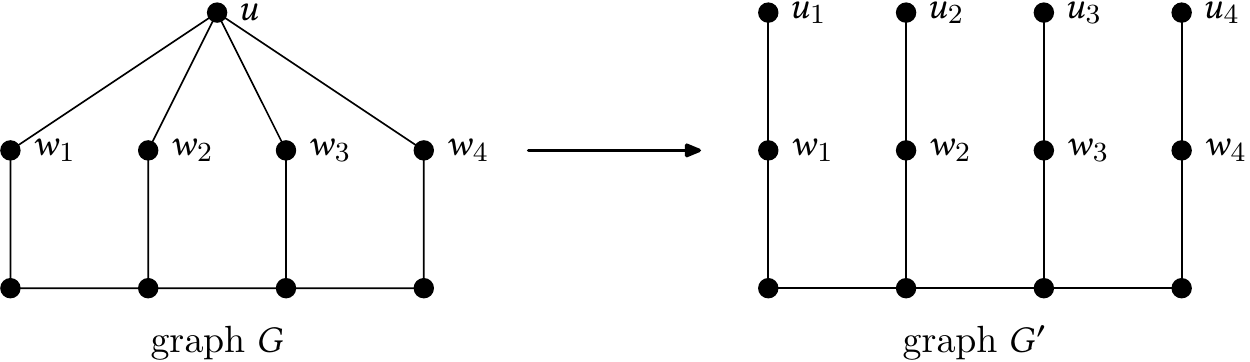}
 \vspace{-3pt}	
 \caption{Split vertex $u$ to modify the graph $G$ to $G'$} \label{figure-ex}
 \end{figure}

For each $0\leq k \leq m$,
we define a set of partial colourings $\sigma_k$ on $\{u_1,u_2,\ldots,u_m\}$ by
\begin{align*}
\sigma_k(u_j) \triangleq \begin{cases}
c_1      &\text{if } 1\leq j \leq k\\
c_2 &\text{if } k+1\leq j \leq m.	
 \end{cases}
\end{align*}
Then, it holds that $\mu_{v,(G,\*L)}^{u \gets c_1} = \mu_{v,(G',\*L')}^{\sigma_m}$ and $\mu_{v,(G,\*L)}^{u \gets c_2} = \mu_{v,(G',\*L')}^{\sigma_0}$.
By the triangle inequality, we have
\begin{align}
\label{eq-triangle}
\DTV{\mu_{v,(G,\*L)}^{u \gets c_1}}	{\mu_{v,(G,\*L)}^{u \gets c_2}} = \DTV{\mu_{v,(G',\*L')}^{\sigma_0}}	{\mu_{v,(G',\*L')}^{\sigma_m}} \leq \sum_{k=1}^{m}	\DTV{\mu_{v,(G',\*L')}^{\sigma_{k-1}}}	{\mu_{v,(G',\*L')}^{\sigma_{k}}}.
\end{align}
We now bound  $\DTV{\mu_{v,(G',\*L')}^{\sigma_{k-1}}}	{\mu_{v,(G',\*L')}^{\sigma_{k}}}$ for each $1\le k\le m$. Consider the following coupling procedure~$\+C$.
\begin{itemize}
\item sample $c,c' \in L'(w_k)=L(w_k)$ from the optimal coupling of $\mu_{w_k,(G',\*L')}^{\sigma_{k-1}}$ and $\mu_{w_k,(G',\*L')}^{\sigma_{k}}$.
\item sample $c_v,c'_v$ from the optimal coupling of $\mu_{v,(G',\*L')}^{\sigma_{k-1},w_k\gets c}$	 and $\mu_{v,(G',\*L')}^{\sigma_{k},w_k\gets c'}$.
\end{itemize}
By the definition of $\sigma_k$ and $\sigma_{k-1}$, they differ only at one vertex $u_k$. 
By the construction of the graph $G'$, $u_k$ is only adjacent to $w_k$. 
Then conditional on the colour of $w_k$, the colour of $u_k$ is independent from the colour of $v$. Hence,
in this coupling, we know that $c_v \neq c'_v$ can happen only if $c \neq c'$. 
Since $c,c'$ are sampled from the optimal coupling, we have $\Pr[\+C]{c \neq c'} = \DTV{\mu_{w_k,(G',\*L')}^{\sigma_{k-1}}}{\mu_{w_k,(G',\*L')}^{\sigma_{k}}}$.
%Since $c,c'$ are also sampled from the optimal coupling,
Therefore,
\begin{align*}
\DTV{\mu_{v,(G',\*L')}^{\sigma_{k-1}}}	{\mu_{v,(G',\*L')}^{\sigma_{k}}} &\leq \Pr[\+C]{c_v \neq c_v'}\\
 &\leq \DTV{\mu_{w_k,(G',\*L')}^{\sigma_{k-1}}}{\mu_{w_k,(G',\*L')}^{\sigma_{k}}}\cdot \max_{c,c'\in L'(w_k)}\DTV{\mu_{v,(G',\*L')}^{\sigma_{k-1},w_k\gets c}}{\mu_{v,(G',\*L')}^{\sigma_{k},w_k\gets c'}}.
\end{align*}

Recall that the graph $G_u$ is obtained by removing $u$ from $G$, and the colour lists $\*L_{u,k}^{c_1,c_2}$ is defined in~\eqref{eq-def-new-list}. We can further derive
%Note that $\sigma_k$ and $\sigma_{k-1}$ disagree only at vertex $w_k$, and $w_k$ is only adjacent to $u_k$ in graph $G'$. We have
\begin{align}
  \DTV{\mu_{v,(G',\*L')}^{\sigma_{k-1}}}	{\mu_{v,(G',\*L')}^{\sigma_{k}}}
  &\leq 	\DTV{\mu_{w_k,(G',\*L')}^{\sigma_{k-1}}}{\mu_{w_k,(G',\*L')}^{\sigma_{k}}}\cdot\max_{c,c'\in L'(w_k)}\DTV{\mu_{v,(G',\*L')}^{\sigma_{k-1},w_k\gets c}}{\mu_{v,(G',\*L')}^{\sigma_{k},w_k\gets c'}}\notag\\
  \tag{$\star$}&=\DTV{\mu_{w_k,(G',\*L')}^{\sigma_{k-1}}}{\mu_{w_k,(G',\*L')}^{\sigma_{k}}}\cdot \max_{c,c'\in L_{u,k}^{c_1,c_2}(w_k)}\DTV{\mu_{v,(G_u,\*L_{u,k}^{c_1,c_2})}^{w_k\gets c}}{\mu_{v,(G_u,\*L_{u,k}^{c_1,c_2})}^{w_k\gets c'}}\notag\\
  \label{eq-couple-colouring}  &=\DTV{\mu_{w_k,(G',\*L')}^{\sigma_{k-1}}}{\mu_{w_k,(G',\*L')}^{\sigma_{k}}}\cdot R_{G_u, \*L_{u,k}^{c_1,c_2}}(w_k,v).
\end{align}
Equation $(\star)$ holds due to $L'(w_k)=L(w_k) = L_{u,k}^{c_1,c_2}(w_k)$ and the definitions of instances $(G',\*L')$ and $(G_u, \*L_{u,k}^{c_1,c_2})$.

Now, our task is reduced to bound $\DTV{\mu_{w_k,(G',\*L')}^{\sigma_{k-1}}}{\mu_{w_k,(G',\*L')}^{\sigma_{k}}}$. 
%Note that the original instance $(G,\*L) \in \^L$. By definition, $(G',\*L') \in \mathscr{F}(\delta,\chi)$. 
%Let $S = \{u_1,u_2,\ldots,u_m\} \setminus \set{u_k}$. Let $\tau = \sigma_k(S) = \sigma_{k-1}(S)$. Namely, $\tau(w_j) = c_1$ for all $1\leq  j \leq k - 1$, and $\tau(w_j) = c_2$ for all $k+1 \leq j \leq m$.
 Let $S$ denote $\set{u_1,\dots,u_m}$. We define two list colouring instances $(G^*_1,\*L^*_1) = \pin_{G',\*L'}(S,\sigma_{k-1})$ and $(G^*_2,\*L^*_2) = \pin_{G',\*L'}(S,\sigma_{k})$. Then we have $\mu_{w_k,(G',\*L')}^{\sigma_{k-1}} = \mu_{w_k,(G^*_1,\*L^*_1)}$ and  $\mu_{w_k,(G',\*L')}^{\sigma_{k}} = \mu_{w_k,(G^*_2,\*L^*_2)}$. Thus
 \begin{align}
\label{eq-dtv-reduction}
\DTV{\mu_{w_k,(G',\*L')}^{\sigma_{k-1}}}{\mu_{w_k,(G',\*L')}^{\sigma_{k}}} = \DTV{\mu_{w_k,(G^*_1,\*L^*_1)}}{\mu_{w_k,(G^*_2,\*L^*_2)}}.	
\end{align}
Besides, $G^*_1=G^*_2=G_u$ and both $(G^*_1,\*L^*_1), (G^*_2,\*L^*_2)$ can be obtained from $(G,\*L)$ by removing $u$ and removing certain colours from $L(u_k)$ for $k=1,\dots,m$. 
So we have that $(G^*_1,\*L^*_1),(G^*_2,\*L^*_2)\in\^L$ since $\^L$ is downward closed. 
Moreover, the two collections of colour lists $\*L_1^*=\set{L_1^*(v)\mid v\in V\setminus\set{u}}$ and $\*L_2^*=\set{L_2^*(v)\mid v\in V\setminus\set{u}}$ can only differ at $w_k$ where $L_1^*(w_k)=L(w_k)\setminus\set{c_2}$ and $L_2^*(w_k)=L(w_k)\setminus\set{c_1}$.

	We prove an auxiliary lemma.
\begin{lemma}
\label{lemma-single-disagree}
Let $\deg_G(w_k)$ denote the degree of $w_k$ in $G$.
\[
\DTV{\mu_{w_k,(G^*_1,\*L^*_1)}}{\mu_{w_k,(G^*_2,\*L^*_2)}}\le \min \tp{\frac{\eps_1}{\deg_{G}(w_k) - 1}, \frac{1}{\eps_2\chi+1}}.
\]
%Let $\mathscr{L}$ be a downward closed family of list colouring instances satisfying \Cref{cond:main} with parameters $\Delta>0$ and $0<\eps_1,\eps_2<1$. Let $(G=(V,E),\*L) \in \mathscr{L}$ be a list colouring instance. For any edge $\set{u,w} \in E$ such that $\deg_G(u) = 1$, it holds that
%\begin{align}
%\label{eq-couple-marginal}
%\forall c_1,c_2 \in L(u),\quad
%\DTV{\mu_{w,(G,\*L)}^{u \gets c_1}}{\mu_{w,(G,\*L)}^{u \gets c_2}} \leq \min \tp{\frac{\eps_1}{\deg_{G}(w) - 1}, \frac{1}{\eps_2\Delta+1}}.
%\end{align}
\end{lemma}
\begin{proof}

It suffices to prove that 
\begin{align}
\label{eq-sr}
\DTV{\mu_{w_k,(G^*_1,\*L^*_1)}}{\mu_{w_k,(G^*_2,\*L^*_2)}} = \max\left\{ \mu_{w_k,(G^*_1,\*L^*_1)}(c_1),\mu_{w_k,(G^*_2,\*L^*_2)}(c_2) \right\}.
\end{align}
%To see that~\eqref{eq-sr} implies the lemma, define two list colouring instances: 
%\begin{itemize}
%\item $\+L_1=(G_u, \*L_1)= \pin_{G,\*L}(\{u\},c_1)$, where $G_u$ is the subgraph of $G$ induced by $V \setminus \{u\}$, $\*L_1 = \{L_1(w) \mid w \in V \setminus \{u\}\}$, and $\pin_{\cdot}(\cdot)$ is defined in \Cref{definition-pin}; 
%\item $\+L_2=(G_u, \*L_2)= \pin_{G,\*L}(\{u\},c_2)$, where $\*L_2 = \{L_2(w) \mid w \in V \setminus \{u\}\}$.
%\end{itemize}
%Since $\mathscr{L}$ is downward closed, we have both $\+L_1,\+L_2 \in \mathscr{L}$. Also note that $\mu_{w,(G,\*L)}^{u \gets c_1} = \mu_{w,\+L_1}$ and  $\mu_{w,(G,\*L)}^{u \gets c_2} = \mu_{w,\+L_2}$. 
To see that \eqref{eq-sr} implies the lemma, note that $(G,\*L) \in \^L$, thus $\deg_{G_u}(w_k) = \deg_G(w_k) - 1 \leq \chi - 1$. 
Since $(G^*_1,\*L^*_1),(G^*_2,\*L^*_2)\in\^L$, \Cref{cond:main} gives
\[
\max\set{ \mu_{w_k,(G^*_1,\*L^*_1)}(c_1),\mu_{w_k,(G^*_2,\*L^*_2)}(c_2)}
\leq \min \tp{\frac{\eps_1}{\deg_{G_u}(w_k)}, \frac{1}{\eps_2\chi+1}} = \min \tp{\frac{\eps_1}{\deg_G(w_k)-1}, \frac{1}{\eps_2\chi+1}}.
\]

It remains to verify~\eqref{eq-sr}. 
Note that %$(G^*_1,\*L^*_1),(G^*_2,\*L^*_2)\in\^L$,
assuming \Cref{cond:main}, the distributions in~\eqref{eq-sr} are well-defined.

Let $(G_u,\*{\widetilde L})$ be a list colouring instance where $\*{\widetilde L}=\set{\widetilde L(v)\mid v\in V\setminus\set{u}}$ differs from $\*L^*_1$ and $\*L^*_2$ only on~$w_k$,
and $\widetilde L(w_k)=L(w_k)$. 
For each colour $c$, define $n(c)$ as the number of proper list colourings of  $(G_u,\*{\widetilde L})$ such that the colour of $w_k$ is $c$. 
Note that $n(c) = 0$ if $c \not\in \widetilde L(w_k)$. Define
\begin{align*}
  N \defeq \sum_{c \in \widetilde L(w_k) \setminus \{c_1,c_2\}}n(c).	
\end{align*}
We claim that
\begin{align}
\label{eq-dtv-max}
\DTV{\mu_{w_k,(G^*_1,\*L^*_1)}}{\mu_{w_k,(G^*_2,\*L^*_2)}} = \frac{\max\{n(c_1),n(c_2)\}}{N + \max\{n(c_1),n(c_2)\}}.
\end{align}
This implies \eqref{eq-sr} as the RHS of~\eqref{eq-dtv-max} equals $\max\left\{ \mu_{w_k,(G^*_1,\*L^*_1)}(c_1),\mu_{w_k,(G^*_2,\*L^*_2)}(c_2) \right\}$. 
To show \eqref{eq-dtv-max}, we may assume $n(c_1) \geq n(c_2)$ first.
Then,
\begin{align*}
\DTV{\mu_{w_k,(G^*_1,\*L^*_1)}}{\mu_{w_k,(G^*_2,\*L^*_2)}} &= \frac{1}{2}\tp{\sum_{c \in L(w_k) \setminus \{c_1,c_2\}}\left\vert\frac{n(c)}{N+n(c_1)} - \frac{n(c)}{N+n(c_2)} \right\vert	+ \frac{n(c_1)}{N+n(c_1)} + \frac{n(c_2)}{N+n(c_2)}}\\
\tag{as $n(c_1)\geq n(c_2)$}&=\frac{1}{2}\tp{ \frac{N(n(c_1)-n(c_2))}{(N+n(c_1))(N+n(c_2))}+ \frac{n(c_1)N+n(c_2)N+2n(c_1)n(c_2)}{(N+n(c_1))(N+n(c_2))} }\\
&=\frac{n(c_1)}{N+n(c_1)}.
\end{align*}
Similarly, $\DTV{\mu_{w_k,(G^*_1,\*L^*_1)}}{\mu_{w_k,(G^*_2,\*L^*_2)}} = \frac{n(c_2)}{N+n(c_2)}$ if $n(c_2) > n(c_1)$.	
This shows \eqref{eq-dtv-max}.
\end{proof}

%Remark that in~\eqref{eq-couple-marginal}, if $\deg_G(w) = 1$, then the RHS = $ \frac{1}{\eps_2\Delta+1}$.
%The proof pf \Cref{lemma-single-disagree} is deferred to \Cref{section-proof-discre}.
%We now use \Cref{lemma-single-disagree} on instance $(G^*,\*L^*)$. 
Combining~\eqref{eq-dtv-reduction} and \Cref{lemma-single-disagree}, we have
\begin{align}
\DTV{\mu_{u_k,(G',\*L')}^{\sigma_{k-1}}}{\mu_{u_k,(G',\*L')}^{\sigma_{k}}} &= \DTV{\mu_{w_k,(G^*_1,\*L^*_1)}}{\mu_{w_k,(G^*_2,\*L^*_2)}}\notag\\
\tag{by \Cref{lemma-single-disagree}}&\leq \min \tp{\frac{\eps_1}{\deg_{G}(w_k) - 1}, \frac{1}{\eps_2\chi+1}}\notag\\
\label{eq-reduction-bound} &= 	 \min \tp{\frac{\eps_1}{\deg_{G_u}(w_k)}, \frac{1}{\eps_2\chi+1}},
\end{align}
where $G_u$ is the subgraph of $G$ induced by $V \setminus \{u\}$. By~\eqref{eq-triangle},~\eqref{eq-couple-colouring} and~\eqref{eq-reduction-bound}, we have
\begin{align*}
R_{G,\*L}(u, v)  \leq 
\sum_{k = 1}^{\deg_G(u)}\min \tp{\frac{\eps_1}{\deg_{G_u}(w_k)}, \frac{1}{\eps_2\chi+1}} \cdot R_{G_u, \*L_{u,k}^{c_1,c_2}}(w_k,v).
\end{align*}
This proves \Cref{lemma-recursive}.

\subsection{Verify \texorpdfstring{\Cref{cond:main}}{Condition 5.2} (Proof of \texorpdfstring{\Cref{thm-2D+1.76D}}{Theorem 5.1})}
\label{section-proof-colouring}

We first introduce the following lemma.
%\begin{lemma}\label{lem-2D}
%	Let $(G=(V,E),\*L)$ be an instance of list colouring where $\*L=\set{L(v)\mid v\in V}$. Let $\Delta>1$ be the maximum degree of $G$. Assume for every $v\in V$, it holds that
%	\[
%		\abs{L(v)}-\deg_G(v)\ge (1+\delta)\Delta.
%	\]
%	Let $\^L$ be the downward closure of $(G,\*L)$. Then $\^L$ satisfies \Cref{cond:main} with parameters $\chi=\Delta$, $\eps_1=(1+\delta)^{-1}$ and $\eps_2 =\frac{1}{2}+\delta$.
%
%\end{lemma}

\begin{lemma}\label{lem-1.76D}
	Let $(G=(V,E),\*L)$ be an instance of list colouring where $G$ is triangle-free and $\*L=\set{L(v)\mid v\in V}$. Let $\Delta\geq 3$ be the maximum degree of $G$ and $\delta>0$ be a constant. Assume for every $v\in V$, it holds that
	\[
		\abs{L(v)}-\deg_G(v)\ge (\alpha^*+\delta-1)\Delta.
	\]
	Let $\^L$ be the downward closure of $(G,\*L)$. Then $\^L$ satisfies \Cref{cond:main} with parameters $\chi=\Delta$, $\eps_1=1-\frac{\delta}{\alpha^*+\delta}$ and $\eps_2=0.4+\delta$.
\end{lemma}

It is clear that \Cref{thm-2D+1.76D} is a consequence of \Cref{lem-1.76D} and \Cref{thm-family-main}.

\begin{proof}[Proof of \Cref{thm-2D+1.76D}]
Suppose the instance $(G,\*L)$ satisfies the condition in~\eqref{eq-proof-1.7D}. 
By \Cref{lem-1.76D}, the downward closure $\^L$ of $(G,\*L)$ satisfies \Cref{cond:main} with parameters $\chi=\Delta$, $\eps_1=1-\frac{\delta}{\alpha^*+\delta}$ and $\eps_2=0.4+\delta$.
By \Cref{thm-family-main}, we have
\begin{align*}
	\tmix(\eps)\le  \tp{9\e^{\frac{2}{\eps_2}}}^{\tp{1+\frac{1}{(1-\eps_1)\eps_2}}} n^{1+\frac{2}{(1-\eps_1)\eps_2}}	  \cdot\log\tp{\frac{M}{\eps}}.
\end{align*}
Note that $\frac{2}{\epsilon_2} = \frac{2}{0.4+\delta} \leq 5$ and  $\frac{1}{(1-\eps_1)\eps_2} = \frac{\alpha^* + \delta}{\delta(0.4+\delta)} \leq \frac{1}{\delta} \cdot\frac{\alpha^*}{0.4} \leq \frac{9}{2\delta}$. Thus, we have
\begin{align*}
	\tmix(\eps)&\le \tp{9\e^5}^{\tp{1+\frac{9}{2\delta}}} n^{1+\frac{9}{\delta}}	  \cdot\log\tp{\frac{M}{\eps}} \leq  \tp{9\e^5 n}^{1 + 9/\delta}\cdot\log\tp{\frac{M}{\eps}}.	\qedhere
\end{align*}
\end{proof}

%\subsection{Proof of \Cref{lem-2D}}
%
%We first claim that every instance $(G=(V,E),\*L=\set{L(v)\mid v\in V})\in\^L$ satisfies 
%\begin{equation}\label{eq-2D}
%  \forall v\in V: \abs{L(v)}-\deg_G(v)\ge (1+\delta)\chi
%\end{equation}
%To see this, we only need to notice that \eqref{eq-2D} is preserved by the $\preceq$ relation, namely if $(G',\*L')$ satisfies \eqref{eq-2D} and $(G,\*L)\preceq (G',\*L')$, then $(G,\*L)$ satisfies \eqref{eq-2D} as well. This holds since by the definition of $\preceq$, $(G,\*L)$ can be obtained from $(G',\*L')$ by removing some vertex $v$ and removing at most one colour from the colour lists of $v$'s neighbours. Therefore, once the size of the colour list of certain vertex $u$ decreases by one, its degree must decrease by one as well. So the LHS of \eqref{eq-2D} never decreases.
%
%	We fix a list colouring instance $(G=(V, E),\*L)\in\^L$. We shall prove  $\mu_{v,(G,\*L)}(c) \leq \frac{1}{(1+\delta)\chi}$, so we can pick $\eps_1=(1+\delta)^{-1}$ and $\eps_2=\frac{1}{2}+\delta$ since $\chi\ge 2$. Conditional on any colouring of $\Gamma_G(v)$, vertex $v$ has at least $(1+\delta)\chi$ available colours and therefore the marginal probability is at most $\frac{1}{(1+\delta)\chi}$. Since $\mu_{v,(G,\*L)}(c)$ is a convex combination of these conditional probabilities, the upper bound follows.

\subsubsection{Proof of \texorpdfstring{\Cref{lem-1.76D}}{}}
We first remark that $\chi \geq 3$.
We then claim that every instance $(G=(V,E),\*L=\set{L(v)\mid v\in V})\in\^L$ satisfies 
\begin{equation}\label{eq-2D}
  \forall v\in V: \abs{L(v)}-\deg_G(v)\ge (\alpha^*+\delta-1)\chi
\end{equation}
and $G$ is triangle-free. 
To see this, we only need to notice that \eqref{eq-2D} is preserved by the $\preceq$ relation, namely if $(G',\*L')$ satisfies \eqref{eq-2D} and $(G,\*L)\preceq (G',\*L')$, then $(G,\*L)$ satisfies \eqref{eq-2D} as well. This holds since by the definition of $\preceq$, $(G,\*L)$ can be obtained from $(G',\*L')$ by removing some vertex $v$ and removing at most one colour from the colour lists of $v$'s neighbours. Therefore, once the size of the colour list of certain vertex $u$ decreases by one, its degree must decrease by one as well. So the LHS of \eqref{eq-2D} never decreases. Besides, it is easy to see all graphs in $\^L$ are triangle-free.

By~\eqref{eq-2D} and $\chi \geq 3$, for any $(G,\*L)\in \^L$,  $\abs{L(v)} \geq \deg_G(v) + 3(\alpha^*+ \delta - 1) \geq \deg_G(v) + 2$ for any vertex $v$. One can construct a proper list colouring using a simple greedy procedure. Hence, a proper list colouring exists for any instance in $\^L$.

	We fix a list colouring instance $(G=(V, E),\*L)\in\^L$. We first prove 
	\begin{align}
	\label{eq-up-1}
		\mu_{v,(G,\*L)}(c) \leq \frac{1}{(\alpha^*+\delta-1)\chi} \overset{(\star)}{\leq} \frac{1}{(0.4+\delta)\chi + 1},
	\end{align}
	where $(\star)$ holds due to $\chi \geq 3$,
	so we can pick $\eps_2=0.4+\delta$. Conditional on any colouring of $\Gamma_G(v)$, vertex $v$ has at least $(\alpha^*+\delta-1)\chi$ available colours and therefore the marginal probability is at most $\frac{1}{(\alpha^*+\delta-1)\chi}\le \frac{1}{(0.4+\delta)\chi + 1}$. Since $\mu_{v,(G,\*L)}(c)$ is a convex combination of these conditional probabilities, the upper bound follows.

Next, fix a vertex $v \in V$ with $\deg_G(v) \leq \chi - 1$. We prove  $\mu_{v,(G,\*L)}(c) \leq \frac{1-\delta/(\alpha^*+\delta)}{\deg_G(v)}$, so we can pick $\eps_1=1-\frac{\delta}{\alpha^*+\delta}$. Let $\Gamma^+_G(v) = \Gamma_G(v) \cup \{v\}$ denote inclusive neighbourhood of $v$. We show that, conditional on any colouring $\sigma$ of $V \setminus \Gamma^+_G(v)$, the marginal probability $\mu_{v,(G,\*L)}^\sigma(c) \leq \frac{1-\delta/(\alpha^*+\delta)}{\deg_G(v)}$. Define a new instance $(\widetilde{G},\widetilde{\*L})=\pin_{G,\*L}(V \setminus \Gamma^+_G(v), \sigma)$, where $\pin_{\cdot}(\cdot)$ is in \Cref{definition-pin}. Since $\^L$ is downward closed, $(\widetilde{G},\widetilde{\*L}) \in \^L$. Let $m = \deg_G(v) = \deg_{\widetilde{G}}(v)$.
It suffices to prove that 
\begin{align}
\label{eq-proof-target}
\forall c \in L(v) = \widetilde{L}(v), \quad \mu_{v,(\widetilde{G},\widetilde{\*L})}(c) \leq \frac{1-\delta/(\alpha^*+\delta)}{m}.	
\end{align}

Note that if  $m = 0$,~\eqref{eq-proof-target} holds trivially. If $ m = 1$ or $m = 2$, by $(\widetilde{G},\widetilde{\*L}) \in \^L$, $\chi \geq 3$ and ~\eqref{eq-up-1}, we have
\begin{align*}
\mu_{v,(\widetilde{G},\widetilde{\*L})}(c) \leq \frac{1}{(\alpha^* + \delta - 1)\chi} \leq \frac{1}{3(\alpha^*+\delta-1)}\overset{(\star)}{\leq} \frac{1-\delta/(\alpha^*+\delta)}{2}\leq \frac{1-\delta/(\alpha^*+\delta)}{m},
\end{align*}
where $(\star)$ holds because $\frac{1}{3(\alpha^*+\delta-1)}  \leq \frac{\alpha^*}{2(\alpha^*+\delta)}$ for all $\delta > 0$.

Now, we assume $m \geq 3$.
Let $v_1,v_2,\ldots,v_m$ denote the neighbours of $v$ in $\widetilde{G}$. For each $1\leq i \leq m$, define $s_i = |\widetilde{L}(v_i)|$, and for any colour $b$, let $\delta_{i,b} =1 $ if $b \in \widetilde{L}(v_i)$; and $\delta_{i,b} =0 $ if $b \not\in \widetilde{L}(v_i)$. Since $\widetilde{G}$ is a triangle-free graph, we have for any $\forall c \in L(v) = \widetilde{L}(v)$,
\begin{align}
\label{eq-upper-proof}
\mu_{v,(\widetilde{G},\widetilde{\*L})}(c) = \frac{\prod_{i=1}^m(s_i-\delta_{i,c})}{\sum_{b \in L(v)}\prod_{i=1}^m(s_i-\delta_{i,b})} = \frac{\prod_{i=1}^m\tp{1-\frac{\delta_{i,c}}{s_i}}}{\sum_{b \in L(v)}\prod_{i=1}^m\tp{1-\frac{\delta_{i,b}}{s_i}}} \leq  \frac{1}{\sum_{b \in L(v)}\prod_{i=1}^m\tp{1-\frac{\delta_{i,b}}{s_i}}}.
\end{align}
We give a lower bound for denominator. Let $s_v = \abs{L(v)}$. By the AM-GM inequality, we have
\begin{align}
\label{eq-lower-proof}
\sum_{b \in L(v)}\prod_{i=1}^m\tp{1-\frac{\delta_{i,b}}{s_i}} \geq s_v \tp{ \prod_{b \in L(v)}\prod_{i=1}^m\tp{1-\frac{\delta_{i,b}}{s_i}} }^{1/s_v}
= s_v \tp{ \prod_{i=1}^m \prod_{b \in L(v) \cap \widetilde{L}(v_i)}\tp{1-\frac{1}{s_i}}}^{1/s_v},
\end{align}
where the last equality holds because $\delta_{i,b} = 1$ if and only if $b \in \widetilde{L}(v_i)$. 
Note that $(L(v) \cap \widetilde{L}(v_i)) \subseteq \widetilde{L}(v_i)$ and $s_i = |\widetilde{L}(v_i)|$, which implies $|L(v) \cap \widetilde{L}(v_i)| \leq s_i$. We have that
\begin{align*}
\prod_{i=1}^m \prod_{b \in L(v) \cap \widetilde{L}(v_i)}\tp{1-\frac{1}{s_i}} \geq 	 \prod_{i=1}^m \tp{1-\frac{1}{s_i}}^{s_i}.
\end{align*}
Let 
$p = (\alpha^* + \delta - 1)m + 0.5$.	 
Since $(\widetilde{G},\widetilde{\*L}) \in \^L$ and $m = \deg_G(v) \leq \chi - 1$, for all $1\leq i \leq m$,
\begin{align*}
s_i \geq (\alpha^* + \delta - 1)\chi \geq (\alpha^* + \delta - 1)(m+1) \geq  (\alpha^* + \delta - 1)m + 0.5 = p.	
\end{align*}
Note that $p > 1$ because $m \geq 3$.
Also note that $f(x) = (1-1/x)^x$ is increasing when $x \geq 1$. Then we have
$\prod_{i=1}^m \tp{1-\frac{1}{s_i}}^{s_i} \geq 	\tp{1-\frac{1}{p}}^{mp}$.
By~\eqref{eq-lower-proof}, we have
\begin{align*}
\sum_{b \in L(v)}\prod_{i=1}^m\tp{1-\frac{\delta_{i,b}}{s_i}}  \geq s_v \tp{1-\frac{1}{p}}^{\frac{mp}{s_v}}.
%\geq s_v		\tp{1-\frac{1}{p}}^{\frac{mp}{s_v}} = s_v \tp{1-\frac{1}{p}}^{(p-1)\frac{mp}{s_v(p-1)}} \geq s_v\exp\tp{-\frac{mp}{s_v(p-1)}},
\end{align*}
%where the last inequality holds because $(1-1/p)^{p-1} \geq 1/\mathrm{e}$. 

Since $(\widetilde{G},\widetilde{\*L}) \in \^L$ and $m = \deg_G(v) \leq \chi - 1$, $s_v \geq m + (\alpha^*+\delta-1)\chi \geq m + (\alpha^*+\delta-1)(m+1) \geq m+p$ . By the fact that $p>1$, we have $\frac{1}{s_v} \leq \frac{1}{m+p}$ and $\tp{1-\frac{1}{p}}^{-\frac{mp}{s_v}} \leq \tp{1-\frac{1}{p}}^{-\frac{mp}{m+p}}$. By~\eqref{eq-upper-proof}, we have
\begin{align}
\label{eq-upper-proof-2}
\mu_{v,(\widetilde{G},\widetilde{\*L})}(c) \leq \frac{1}{s_v}\tp{1-\frac{1}{p}}^{-\frac{mp}{s_v}}
\le \frac{1}{m+p}\tp{1-\frac{1}{p}}^{-\frac{mp}{m+p}}.
\end{align}

To proof~\eqref{eq-proof-target}, we define the following function
\begin{equation*}
f(m) \triangleq \frac{m+p}{m}\tp{1-\frac{1}{p}}^{\frac{mp}{m+p}} = \frac{(\alpha^*+\delta)m + 0.5}{m}\tp{1-\frac{1}{(\alpha^* + \delta - 1)m + 0.5}}^{\frac{m((\alpha^* + \delta - 1)m + 0.5)}{(\alpha^* + \delta )m + 0.5}}	 
\end{equation*}
By definition, $\mu_{v,(\widetilde{G},\widetilde{\*L})}(c) \leq \frac{1}{mf(m)}$.
In \Cref{lemma-decrasing}, we show that $f(m)$ is a decreasing function for $m \geq 3$. Thus, we have
\begin{align*}
f(m) \geq \lim_{x \rightarrow \infty} f(x)	= (\alpha^*+\delta)\exp\tp{-\frac{1}{\alpha^*+\delta}}.
\end{align*}
Thus, we have
\begin{align*}
\mu_{v,(\widetilde{G},\widetilde{\*L})}(c) \leq \frac{1}{mf(m)} \leq \frac{1}{m}\cdot	\frac{1}{\alpha^*+\delta}\exp\tp{\frac{1}{\alpha^*+\delta}} \overset{(\star)}{\le} \frac{1}{m}\cdot\frac{\alpha^*}{\alpha^*+\delta} = \frac{1-\delta/(\alpha^*+\delta)}{m}.
\end{align*}
where $(\star)$ is due to the fact that $\exp\tp{\frac{1}{\alpha^*}}=\alpha^*$. This proves~\eqref{eq-proof-target} for all $m \geq 3$.

%Note that the LHS of~\eqref{eq-mp} is increasing in $m$ and $m = \deg_{\widetilde{G}}(v) \leq \chi$, we obtain
%\begin{align*}
%\frac{m}{m+p}\exp\tp{\frac{mp}{(m+p)(p-1)}} 
%&\le \frac{\chi}{\chi+p}\exp\tp{\frac{\chi p}{(\chi+p)(p-1)}}\\
%&=\frac{\chi}{(\alpha^*+\delta)\chi+3}\exp\tp{\frac{\chi}{(\alpha^*+\delta)\chi+3}\cdot\frac{(\alpha^*+\delta-1)\chi+3}{(\alpha^*+\delta-1)\chi+2}}
%\end{align*}
%
%We let $g(\chi)\defeq \frac{\chi}{(\alpha^*+\delta)\chi+3}\cdot\frac{(\alpha^*+\delta-1)\chi+3}{(\alpha^*+\delta-1)\chi+2}$. Recall that $\alpha^*>1.763$ and therefore direct calculation yields
%\[
%	g'(\chi)=\frac{18+12(\alpha^*+\delta-1)\chi+(3-5(\alpha^*+\delta)+2(\alpha^*+\delta)^2)\chi^2}{(3+(\alpha^*+\delta)\chi)^2(2+(\alpha^*+\delta-1)\chi)^2}>0
%\]
%for any $\chi>0$. Let $f(\chi) = \frac{\chi}{(\alpha^*+\delta)\chi + 3}\exp(g(\chi))$. Then $f(\chi)$ is increasing when $\chi > 0$. This implies $f(\chi) \leq \lim_{x \to \infty}f(x) = \frac{1}{\alpha^*+\delta}\exp{\tp{\frac{1}{\alpha^*+\delta}}}$ for all $\chi > 0$. Therefore,
%%Therefore, $g(x) \leq \lim_{y \to \infty}g(y) = $, we have 
%\[
%\frac{m}{m+p}\exp\tp{\frac{mp}{(m+p)(p-1)}} \le \frac{1}{\alpha^*+\delta}\exp\tp{\frac{1}{\alpha^*+\delta}}\overset{(\star)}{\le} \frac{\alpha^*}{\alpha^*+\delta}= 1-\frac{\delta}{\alpha^*+\delta},
%\]
%where $(\star)$ is due to the fact that $\exp\tp{\frac{1}{\alpha^*}}=\alpha^*$. This proves~\eqref{eq-mp}.

\subsection{Tightness of the marginal upper bound}
Our whole analysis relies on the upper bound of the marginal probabilities, 
which states that the marginal probability of $v$ taking a specific colour is less than the reciprocal value of $v$'s degree.
%Similar property was also used in the analysis of \cite{liu2019deterministic} (for a very different purpose).
Similar properties were also used in analysing strong spatial mixing~\cite{GMP05,GKM15} or zero-free regions~\cite{liu2019deterministic} for graph colourings.
A natural question is whether the upper bound can be further improved.

%In this section, we show that the bound we obtained to establish \Cref{cond:main} is tight, 
%namely if we replace $\alpha^*$ in the statement of \Cref{thm-2D+1.76D} by any $\alpha<\alpha^*$, 
%\Cref{cond:main} no longer holds. 
%This means that the bound in \Cref{thm-2D+1.76D} is the best we can achieve using current techniques.
%However, we also remark that the construction below only applies to list colouring instances.
%

In this section, we show that the bound $\abs{L(v)} > \alpha^* \Delta \pm O(1) $ is tight for our technique based on the upper bound of the marginal probabilities.
This means that the bound in \Cref{thm-2D+1.76D} is the best we can achieve using current techniques.
However, we also remark that the construction below only applies to list colouring instances.

%we obtained to establish \Cref{cond:main} is tight, 
%namely if we replace $\alpha^*$ in the statement of \Cref{thm-2D+1.76D} by any $\alpha<\alpha^*$, 
%\Cref{cond:main} no longer holds. 
%This means that the bound in \Cref{thm-2D+1.76D} is the best we can achieve using current techniques.
%However, we also remark that the construction below only applies to list colouring instances.
%

We show that there exists a list colouring instance $(G,\*L)$ with triangle-free $G$ such that if $\abs{L(v)} <\alpha^*\Delta - 3$ for some vertex $v$, then $(G,\*L)$ does not have the desired marginal upper bound.
Consider a star $G=(V,E)$ with $(\Delta+1)$ vertices, where $V=\{v,v_1,v_2,\ldots,v_{\Delta}\}$ and $E = \{\{v,v_i\}\mid 1\leq i \leq \Delta\}$. Define  colour lists $\*L $ by $L(v) = [q] = \{0,1,\ldots,q-1\}$ and $L(v_i) = [q-1]=\{0,1,\ldots,q-2\}$ for all $1\leq i \leq \Delta$. 
\begin{proposition}
If $q < \alpha^*\Delta - 3$, then $\mu_{v,(G,\*L)}(c) > \frac{1}{\deg_{G}(v)} = \frac{1}{\Delta}$, where $c$ is the colour $q-1$.
\end{proposition}

\begin{proof}
We can calculate the probability that $v$ takes the colour $c = q-1$ as follows
\begin{align*}
\mu_{v,(G_1,\*L_1)}(c)  = \frac{(q-1)^\Delta}{(q-1)(q-2)^\Delta + (q-1)^\Delta } = \frac{1}{(q-2)\left(1-\frac{1}{q-1}\right)^{(\Delta-1)} + 1 } \geq \frac{1}{(q-2)\exp\tp{-\frac{\Delta-1}{q-1}}+1}.
\end{align*}
If $q < \alpha^*\Delta - 3$, we can verify that ${(q-2)\exp\tp{-\frac{\Delta-1}{q-1}}+1} < \Delta$. This proves the proposition.
\end{proof}

Note that the graph $G$ is a tree, which means that no matter how large we assume the girth of the graph to be, such barrier of marginal upper bounds still exists.
%our result is still tight using current techniques.
%

%Next we drop the condition of being triangle-free and show that $q\ge (2+\delta)\Delta$ is the best bound we can obtain via our technique for general graphs. 
%We construct a list colouring instance $(G_2,\*L_2)$ ($G_2$ may contain triangles) such that if $\abs{L_2(v)} <2\Delta$ for some vertex $v$, then $(G_2,\*L_2)$ does not satisfy the marginal upper bound in \Cref{cond:main}.
%Consider a clique $G_2=(V,E)$ with $(\Delta+1)$ vertices, where $V=\{v,v_1,v_2,\ldots,v_{\Delta}\}$ and $E = \{\{u,v\}\mid u,v \in V, u\neq v\}$. Define the colour lists $\*L_2$ by $L_2(v) = [q] = \{0,1,\ldots,q-1\}$ and $L_2(v_i) = [q-1]=\{0,1,\ldots,q-2\}$ for all $1\leq i \leq \Delta$. 
%\begin{proposition}
%If $q < 2\Delta$, then $\mu_{v,(G_2,\*L_2)}(q) > \frac{1}{\deg_{G_2}(v)} = \frac{1}{\Delta}$.
%\end{proposition}
%\begin{proof}
%We can calculate the probability that $v$ takes the colour $c=q-1$ as follows
%\begin{align*}
%\mu_{v,(G_2,\*L_2)}(c)  = \frac{\prod_{i=1}^\Delta(q-i)}{(q-1)\prod_{i=1}^\Delta(q-1-i) + \prod_{i=1}^\Delta(q-i)} = \frac{1}{q-\Delta}.
%%\frac{1}{(q-2)\left(1-\frac{1}{q-1}\right)^{(\Delta-1)} + 1 } \geq \frac{1}{(q-2)\exp\tp{-\frac{\Delta-1}{q-1}}+1}.
%\end{align*}
%If $q < 2\Delta$, then $q-\Delta < \Delta$. This proves the proposition.
%\end{proof}
%
%This shows that our technique cannot obtain a bound better than $q \geq 2 \Delta$ for sampling proper $q$-colourings in general graphs. 
%On the other hand, it is easy to verify that \Cref{cond:main} holds when $q\geq (2+\delta)\Delta$ for any constant $\delta>0$.

Indeed, the upper bound~\eqref{eq-marginal-up-2} in \Cref{cond:main} is only required for vertices $v$ with $\deg_G(v) \leq \chi - 1$, but a simple modification of the instance above can provide a counter example to \Cref{cond:main}. 
Similar barriers of the marginal upper bound also appear in~\cite{GMP05,GKM15,liu2019deterministic}.
Finally, we remark that the barrier discussed in this section only applies for  our current technique, which is  solely based on marginal upper bounds. It may still be possible to improve the dependence between the number of colours and the degree of the graph by exploiting spectral independence (\Cref{def:spectral-independence}) through other means.

\bibliographystyle{alpha}
\bibliography{refs.bib}

\newcommand{\etalchar}[1]{$^{#1}$}
\begin{thebibliography}{CDM{\etalchar{+}}19}

\bibitem[AL20]{alev2020improved}
Vedat~Levi Alev and Lap~Chi Lau.
\newblock Improved analysis of higher order random walks and applications.
\newblock In {\em STOC}, pages 1198--1211, 2020.

\bibitem[ALO20]{anari2020spectral}
Nima Anari, Kuikui Liu, and Shayan {Oveis Gharan}.
\newblock Spectral independence in high-dimensional expanders and applications
  to the hardcore model.
\newblock {\em arXiv preprint arXiv:2001.00303}, 2020.
\newblock FOCS 2020, to appear.

\bibitem[ALOV19]{ALOV19}
Nima Anari, Kuikui Liu, Shayan {Oveis Gharan}, and Cynthia Vinzant.
\newblock Log-concave polynomials {II:} high-dimensional walks and an {FPRAS}
  for counting bases of a matroid.
\newblock In {\em STOC}, pages 1--12, 2019.

\bibitem[ALOV20]{ALOV20}
Nima Anari, Kuikui Liu, Shayan {Oveis Gharan}, and Cynthia Vinzant.
\newblock Log-concave polynomials {IV:} {E}xchange properties, tight mixing
  times, and faster sampling of spanning trees.
\newblock {\em CoRR}, abs/2004.07220, 2020.

\bibitem[BD97]{bubley1997path}
Russ Bubley and Martin Dyer.
\newblock Path coupling: A technique for proving rapid mixing in {Markov}
  chains.
\newblock In {\em Proceedings of the 38th IEEE Annual Symposium on Foundations
  of Computer Science (FOCS)}, pages 223--231, 1997.

\bibitem[CDM{\etalchar{+}}19]{chen2019improved}
Sitan Chen, Michelle Delcourt, Ankur Moitra, Guillem Perarnau, and Luke Postle.
\newblock Improved bounds for randomly sampling colorings via linear
  programming.
\newblock In {\em SODA}, pages 2216--2234, 2019.

\bibitem[CGM19]{CGM19}
Mary Cryan, Heng Guo, and Giorgos Mousa.
\newblock Modified log-sobolev inequalities for strongly log-concave
  distributions.
\newblock In {\em FOCS}, pages 1358--1370, 2019.

\bibitem[Che98]{MFC98}
Mu-Fa Chen.
\newblock Trilogy of couplings and general formulas for lower bound of spectral
  gap.
\newblock In {\em Probability towards 2000 ({N}ew {Y}ork, 1995)}, volume 128 of
  {\em Lect. Notes Stat.}, pages 123--136. Springer, New York, 1998.

\bibitem[CLV20]{chen2020rapid}
Zongchen Chen, Kuikui Liu, and Eric Vigoda.
\newblock Rapid mixing of glauber dynamics up to uniqueness via contraction.
\newblock {\em arXiv preprint arXiv:2004.09083}, 2020.
\newblock FOCS 2020, to appear.

\bibitem[DF01]{dyer2001fewer}
Martin~E. Dyer and Alan~M. Frieze.
\newblock Randomly colouring graphs with lower bounds on girth and maximum
  degree.
\newblock In {\em FOCS}, pages 579--587, 2001.

\bibitem[DFHV13]{dyer2013randomly}
Martin~E. Dyer, Alan~M. Frieze, Thomas~P. Hayes, and Eric Vigoda.
\newblock Randomly coloring constant degree graphs.
\newblock {\em Random Struct. Algorithms}, 43(2):181--200, 2013.

\bibitem[DGU14]{DGU14}
Martin~E. Dyer, Catherine Greenhill, and Mario Ullrich.
\newblock Structure and eigenvalues of heat-bath {M}arkov chains.
\newblock {\em Linear Algebra Appl.}, 454:57--71, 2014.

\bibitem[DK17]{DK17}
Irit Dinur and Tali Kaufman.
\newblock High dimensional expanders imply agreement expanders.
\newblock In {\em {FOCS}}, pages 974--985. {IEEE} Computer Society, 2017.

\bibitem[FV07]{frieze2007survey}
Alan~M. Frieze and Eric Vigoda.
\newblock A survey on the use of markov chains to randomly sample colourings.
\newblock {\em Oxford Lecture Series in Mathematics and its Applications},
  34:53, 2007.

\bibitem[GK12]{GK12}
David Gamarnik and Dmitriy Katz.
\newblock Correlation decay and deterministic {FPTAS} for counting colorings of
  a graph.
\newblock {\em J. Discrete Algorithms}, 12:29--47, 2012.

\bibitem[GKM15]{GKM15}
David Gamarnik, Dmitriy Katz, and Sidhant Misra.
\newblock Strong spatial mixing of list coloring of graphs.
\newblock {\em Random Struct. Algorithms}, 46(4):599--613, 2015.

\bibitem[GL18]{guo2018uniqueness}
Heng Guo and Pinyan Lu.
\newblock Uniqueness, spatial mixing, and approximation for ferromagnetic
  2-spin systems.
\newblock {\em ACM Trans. Comput. Theory}, 10(4):Art. 17, 25, 2018.

\bibitem[GMP05]{GMP05}
Leslie~Ann Goldberg, Russell~A. Martin, and Mike Paterson.
\newblock Strong spatial mixing with fewer colors for lattice graphs.
\newblock {\em {SIAM} J. Comput.}, 35(2):486--517, 2005.

\bibitem[Hay03]{hayes2003randomly}
Thomas~P. Hayes.
\newblock Randomly coloring graphs of girth at least five.
\newblock In {\em STOC}, pages 269--278, 2003.

\bibitem[Hay06]{hayes2006simple}
Thomas~P. Hayes.
\newblock A simple condition implying rapid mixing of single-site dynamics on
  spin systems.
\newblock In {\em Proceedings of the 47th Annual IEEE Symposium on Foundations
  of Computer Science (FOCS)}, pages 39--46, 2006.

\bibitem[Hay13]{hayes2013local}
Thomas~P. Hayes.
\newblock Local uniformity properties for glauber dynamics on graph colorings.
\newblock {\em Random Struct. Algorithms}, 43(2):139--180, 2013.

\bibitem[HJ12]{horn2012matrix}
Roger~A. Horn and Charles~R. Johnson.
\newblock {\em Matrix analysis}.
\newblock Cambridge university press, 2012.

\bibitem[HV03]{hayes2003non}
Thomas~P. Hayes and Eric Vigoda.
\newblock A non-markovian coupling for randomly sampling colorings.
\newblock In {\em FOCS}, pages 618--627, 2003.

\bibitem[HV06]{HV06}
Thomas~P. Hayes and Eric Vigoda.
\newblock Coupling with the stationary distribution and improved sampling for
  colorings and independent sets.
\newblock {\em Ann. Appl. Probab.}, 16, 2006.

\bibitem[Jer95]{jerrum1995very}
Mark Jerrum.
\newblock A very simple algorithm for estimating the number of $k$-colorings of
  a low-degree graph.
\newblock {\em Random Struct. Algorithms}, 7(2):157--165, 1995.

\bibitem[JS89]{jerrum1989approximating}
Mark Jerrum and Alistair Sinclair.
\newblock Approximating the permanent.
\newblock {\em SIAM J. Comput.}, 18(6):1149--1178, 1989.

\bibitem[JVV86]{jerrum1986random}
Mark Jerrum, Leslie~G. Valiant, and Vijay~V. Vazirani.
\newblock Random generation of combinatorial structures from a uniform
  distribution.
\newblock {\em Theoret. Comput. Sci.}, 43:169--188, 1986.

\bibitem[KO20]{KO20}
Tali Kaufman and Izhar Oppenheim.
\newblock High order random walks: {B}eyond spectral gap.
\newblock {\em Combinatorica}, 40(1):245--281, 2020.

\bibitem[LLY13]{LLY13}
Liang Li, Pinyan Lu, and Yitong Yin.
\newblock Correlation decay up to uniqueness in spin systems.
\newblock In {\em {SODA}}, pages 67--84. {SIAM}, 2013.

\bibitem[LP17]{levin2017markov}
David~A Levin and Yuval Peres.
\newblock {\em Markov chains and mixing times}, volume 107.
\newblock American Mathematical Soc., 2017.

\bibitem[LSS19]{liu2019deterministic}
Jingcheng Liu, Alistair Sinclair, and Piyush Srivastava.
\newblock A deterministic algorithm for counting colorings with $2 {\Delta}$
  colors.
\newblock In {\em {FOCS}}, pages 1380--1404, 2019.

\bibitem[LY13]{lu2013improved}
Pinyan Lu and Yitong Yin.
\newblock Improved {FPTAS} for multi-spin systems.
\newblock In {\em RANDOM}, pages 639--654, 2013.

\bibitem[Mol04]{molloy2004glauber}
Michael Molloy.
\newblock The glauber dynamics on colorings of a graph with high girth and
  maximum degree.
\newblock {\em SIAM J. Comput.}, 33(3):721--737, 2004.

\bibitem[Opp18]{Opp18}
Izhar Oppenheim.
\newblock Local spectral expansion approach to high dimensional expanders part
  {I:} {D}escent of spectral gaps.
\newblock {\em Discret. Comput. Geom.}, 59(2):293--330, 2018.

\bibitem[SS97]{salas1997absence}
Jes{\'u}s Salas and Alan~D Sokal.
\newblock Absence of phase transition for antiferromagnetic {Potts} models via
  the {Dobrushin} uniqueness theorem.
\newblock {\em J. Stat. Phys.}, 86(3):551--579, 1997.

\bibitem[SS20]{SS20}
Shuai Shao and Yuxin Sun.
\newblock Contraction: {A} unified perspective of correlation decay and
  zero-freeness of 2-spin systems.
\newblock In {\em {ICALP}}, volume 168 of {\em LIPIcs}, pages 96:1--96:15.
  Schloss Dagstuhl - Leibniz-Zentrum f{\"{u}}r Informatik, 2020.

\bibitem[Str06]{Str06}
Adam~W. Strzebonski.
\newblock Cylindrical algebraic decomposition using validated numerics.
\newblock {\em J. Symb. Comput.}, 41(9):1021--1038, 2006.

\bibitem[Vig00]{vigoda2000improved}
Eric Vigoda.
\newblock Improved bounds for sampling colorings.
\newblock {\em J. Math. Phys.}, 41(3):1555--1569, 2000.

\bibitem[Wei06]{Wei06}
Dror Weitz.
\newblock Counting independent sets up to the tree threshold.
\newblock In {\em {STOC}}, pages 140--149, 2006.

\bibitem[YZ13]{yin2013approximate}
Yitong Yin and Chihao Zhang.
\newblock Approximate counting via correlation decay on planar graphs.
\newblock In {\em SODA}, pages 47--66. SIAM, 2013.

\end{thebibliography}

\appendix
\section{Computer assisted proof}
%We use Mathematica to prove the following lemma.
We give a computer-assisted proof for the following technical lemma used in the analysis for colouring. 
\begin{lemma}
\label{lemma-decrasing}
Let $\alpha^* \approx 1.763\ldots$ be the solution of $\alpha^* = \exp\tp{\frac{1}{\alpha^*}}$ and $\delta > 0$ a real number. Define
\begin{align*}
f(x) = \frac{(\alpha^*+\delta)x + 0.5}{x}\tp{1-\frac{1}{(\alpha^* + \delta - 1)x + 0.5}}^{\frac{x((\alpha^* + \delta - 1)x + 0.5)}{(\alpha^* + \delta )x + 0.5}}	.
\end{align*}
The function $f$ is decreasing for $x \in [3, \infty)$.
\end{lemma}
\begin{proof}
	Let $a=\alpha^*+\delta-1>0.763$. Direct calculation yields $f'(x)=A\cdot B$ with
\begin{align*}
	A & = \tp{2x^2(2ax-1)(2(a+1)x+1)}^{-1}\tp{1-\frac{2}{1+2ax}}^{\frac{x(1+2ax)}{1+2(a+1)x}};\\
	B &=  1 + 2 x - 4 a^2 x^2 + 8 a (1 + a) x^3 + 
 x \tp{-1 + 8 a^3 x^3 - 2 a x (1 + 2 x) + 4 a^2 x^2 (1 + 2 x)} \ln\tp{
   1 - \frac{2}{1 + 2 a x}}.
\end{align*}
It is easy to see that $A>0$, so we only need to verify that $B<0$. To see this, note that the term
\[
-1 + 8 a^3 x^3 - 2 a x (1 + 2 x) + 4 a^2 x^2 (1 + 2 x)=2ax(1+2x)(2ax-1)+\tp{8a^3x^3-1}>0
\]
for any $x\ge 3$ and $a\ge \alpha^*-1$. It follows from the Taylor series of $\ln(1-z)$ that 
\[
	\ln\tp{
   1 - \frac{2}{1 + 2 a x}} \le -\frac{2}{1+2ax}-\frac{2}{(1+2ax)^2}-\frac{8/3}{(1+2ax)^3}-\frac{4}{(1+2ax)^4}.
\]
Therefore we have
\begin{align*}
	B&\le 1 + 2 x - 4 a^2 x^2 + 8 a (1 + a) x^3 \\
	&\quad+ 
 x \tp{-1 + 8 a^3 x^3 - 2 a x (1 + 2 x) + 4 a^2 x^2 (1 + 2 x)} \\
 &\quad\cdot \tp{-\frac{2}{1+2ax}-\frac{2}{(1+2ax)^2}-\frac{8/3}{(1+2ax)^3}-\frac{4}{(1+2ax)^4}}
 \overset{(\star)}{<} 0,
\end{align*}
where $(\star)$ is verified by the following Mathematica code:

\begin{lstlisting}
Resolve[Exists[x,1+2x-4a^2x^2+8a(1+a)x^3+x(-1+8a^3x^3-2a x(1+2x)+4 a^2x^2(1+2x))*(-(2/(1+2a*x))-2/(1+2a*x)^2-(8/3)/(1+2a*x)^3-4/(1+2a*x)^4)>=0 && a>763/1000 && x>=3]]
\end{lstlisting}
\end{proof}

Here we used the Resolve command in Mathematica. 
This is a rigorous implementation of a quantifier elimination algorithm, which determines the feasibility of a collection of polynomial inequalities.
For more details, see \cite{Str06}.
	
\end{document}